\documentclass[11pt]{article}

\usepackage{amsmath,amsthm,amssymb}
\usepackage[margin=1in]{geometry}
\usepackage{graphicx}
\usepackage{hyperref}
\usepackage{mathtools}
\usepackage{enumerate}
\usepackage{enumitem}
\usepackage{titlesec}
\usepackage{datetime}
\usepackage{bbm}
\usepackage{caption}
\usepackage{subcaption}
\usepackage{float}
\usepackage{graphicx}
\usepackage[page]{appendix}
\usepackage{comment}
\usepackage{array}
\usepackage{booktabs}
\usepackage{dsfont}
\usepackage{color}
\usepackage{xcolor}
\usepackage{thmtools}
\usepackage{thm-restate}
\usepackage{regexpatch}

\makeatletter
\xpatchcmd\thmt@restatable{%
  \csname #2\@xa\endcsname\ifx\@nx#1\@nx\else[{#1}]\fi
}{%
  \ifthmt@thisistheone
    \csname #2\@xa\endcsname\ifx\@nx#1\@nx\else[{#1}]\fi
  \else
    \csname #2\@xa\endcsname[\textbf{re-stated}]
  \fi
}{}{}
\makeatother

\newlist{algosteps}{enumerate}{1}
\setlist[algosteps,1]{label=\textbf{Step \arabic*:}, left=0pt, itemsep=0.5ex}

\newcommand{\braket}[1]{\langle #1 \rangle}
\newcommand{\ket}[1]{| #1 \rangle}
\newcommand{\bra}[1]{\langle #1 |}

\definecolor{sammyPurple}{RGB}{128,0,128}

\theoremstyle{definition}
\newtheorem{theorem}{Theorem}[section]

\newtheorem{lemma}[theorem]{Lemma}
\newtheorem{corollary}[theorem]{Corollary}
\newtheorem{conjecture}[theorem]{Conjecture}
\newtheorem{definition}[theorem]{Definition}

\begin{document}

\title{Testing graph isomorphism using simple quantum circuits}
\author{Samuel Beasley\thanks{University of Bristol} ~and Ashley Montanaro\thanks{University of Bristol and Phasecraft Ltd.}}
\date{\today}

\maketitle

\begin{abstract}
    The graph isomorphism problem asks whether two graphs are equivalent up to a permutation of their vertices. Its precise computational complexity remains unresolved, and it is still unknown whether a polynomial-time algorithm exists. We investigate whether simple graph-dependent Instantaneous Quantum Polynomial-time (IQP) and Quantum Approximate Optimisation Algorithm (QAOA) circuits can distinguish non-isomorphic graphs. We first show that suitable IQP and QAOA constructions allow the underlying graph to be reconstructed using a polynomial number of samples. Building on previous work using the QAOA energy, we show that Hamming weight distributions can distinguish graph pairs at circuit depths for which the numerically computed energy gap is zero. We then apply graph invariants to the matrix of ZZ correlators on the output distribution. For the specific case of the graph spectrum, at QAOA depth $p=2$, this invariant numerically distinguishes all $10{,}486$ cospectral regular graph pairs tested on up to 18 vertices. For IQP circuits, we introduce a finer invariant based on probability multiplicities within Hamming weight sectors, which distinguishes every non-isomorphic graph pair tested numerically. Finally, we derive exact graph-theoretic characterisations of the IQP output distribution at $\theta=\pi/2$ and $\theta=\pi/4$. These results show that simple graph-dependent quantum circuits encode substantial structural information, while leaving open whether this information can be extracted efficiently for general graph isomorphism.

\end{abstract}

\section{Introduction}

Two graphs are said to be isomorphic if there exists a permutation of their vertices that preserves adjacency and non-adjacency. The graph isomorphism problem asks whether such a permutation exists for a given pair of graphs. Graph isomorphism and closely related graph-matching problems arise in a range of applications. In cheminformatics, molecular structures are naturally represented as graphs. Isomorphism testing and canonical labelling are used to compare, identify and catalogue chemical compounds~\cite{Faulon1998}. Graph matching is also widely used in pattern recognition and machine vision, where graphs provide a natural representation of relational and structural information~\cite{Conte2004}. Beyond these applications, graph isomorphism has been a long-standing challenge in algorithm design because its precise computational complexity remains unresolved.

For more than three decades, the best known general worst-case running time for graph isomorphism was $\exp\left(O(\sqrt{n\log n})\right)$, obtained from group-theoretic methods developed by Babai, Kantor and Luks~\cite{BabaiKantorLuks1983}. In a major breakthrough, Babai subsequently gave a quasipolynomial-time algorithm with running time $\exp((\log n)^{O(1)})$~\cite{Babai_2016}. Nevertheless, it remains unknown whether graph isomorphism can be solved in polynomial time.

Despite this unresolved worst-case complexity, graph isomorphism and the closely related problems of canonical labelling and automorphism-group computation can often be solved efficiently in practice. Individualisation-refinement methods form the basis of widely used software packages including \textsc{nauty}, \textsc{traces}, \textsc{bliss}, \textsc{saucy} and \textsc{conauto}~\cite{MCKAY201494}. However, specially constructed graph families can be significantly more difficult. In particular, Miyazaki constructed families of coloured, connected 3-regular graphs for which McKay's canonical labelling algorithm, as implemented in \textsc{nauty}, exhibits exponential behaviour~\cite{miyazaki1997complexity}. Such examples illustrate the distinction between strong practical performance on typical instances and the behaviour of graph-isomorphism algorithms on carefully constructed hard cases.

Graph isomorphism has also long been considered as a potential application of quantum computation. The problem can be reduced to an instance of the non-Abelian hidden subgroup problem over a group related to the symmetric group, so an efficient quantum algorithm for the relevant hidden subgroup problem would yield an efficient quantum algorithm for graph isomorphism~\cite{HallgrenRussellTaShma2003}. However, techniques that successfully solve Abelian hidden subgroup problems do not extend straightforwardly to this setting. In particular, strong Fourier sampling over the symmetric group is insufficient to solve graph isomorphism~\cite{MooreRussellSchulman2008}, and more general hidden-subgroup approaches require highly entangled measurements over many coset states~\cite{MooreRussellSniady2010}.

A separate line of research has investigated the use of quantum walks for graph isomorphism~\cite{DouglasWang2008,Gamble2010,Rudinger2012}. These methods can produce graph invariants that distinguish many non-isomorphic graphs, but they are not complete: particular pairs of non-isomorphic strongly regular graphs have been shown to remain indistinguishable under certain quantum-walk constructions~\cite{DouglasWang2008,Gamble2010}. More generally, Rudinger \emph{et al.} showed that non-interacting quantum walks with any fixed number of particles cannot distinguish all strongly regular graphs~\cite{Rudinger2012}. These limitations motivate the investigation of alternative quantum-based methods for extracting graph structure.

More recently, Mario Szegedy proposed the use of the Quantum Approximate Optimisation Algorithm (QAOA), originally constructed by Farhi, Goldstone and Gutmann in 2014~\cite{farhi_2014}, for graph structure discovery and isomorphism testing~\cite{szegedy_2019_qaoa}. Rather than optimising the circuit parameters to solve a combinatorial optimisation problem, Szegedy considered the QAOA energy evaluated at randomly chosen parameter sequences. Since isomorphic graphs produce identical energies for the same QAOA parameters, the energy can be treated as a graph invariant. At depth $p=1$, Szegedy showed that for a 3-regular graph the QAOA energy contains exactly the information given by the number of vertices and the number of triangles. At larger depths, however, Szegedy's numerical results indicate substantially greater distinguishing power: his numerical experiments showed that QAOA energies distinguish all non-isomorphic 3-regular graphs tested up to 18 vertices, all strongly regular graphs tested up to 26 vertices, as well as the Praust and smallest Miyazaki examples.

Szegedy's results also revealed an important limitation of using the QAOA energy alone. Szegedy considered pairs consisting of a circular ladder graph and the corresponding M\"obius ladder graph and found numerically that their average QAOA energy gap appears to decrease exponentially as the size of the graphs increases. Thus, even if two non-isomorphic graphs have different QAOA energies in principle, the separation may become too small to resolve efficiently in practice. This raises a natural question: does the QAOA output distribution contain additional graph information that can distinguish such hard instances at lower circuit depths than the energy alone?

In this work, we investigate this question by studying increasingly fine-grained information contained in the output distributions of simple quantum circuits associated with graphs. We consider both Instantaneous Quantum Polynomial-time (IQP) circuits~\cite{Shepherd_2009} and QAOA circuits~\cite{farhi_2014}. Our main results are as follows.

\begin{enumerate}
\item We first show that suitable IQP and QAOA constructions allow the underlying graph to be reconstructed using a polynomial number of samples. Thus, these circuit families do not discard the information required to determine graph isomorphism.

\item We show that certain graph pairs that cannot be distinguished using QAOA energies at low circuit depth can nevertheless be distinguished using other information contained in the output distribution. In particular, the Hamming weight distribution{, which can be estimated efficiently,} distinguishes several such pairs at depths for which the numerically computed QAOA energy gap is zero.

\item We introduce the ZZ correlation matrix $M_G$ obtained from the QAOA distribution. Under a permutation of the graph vertices, this matrix undergoes the corresponding simultaneous permutation of its rows and columns. Consequently, any function of $M_G$ that is invariant under such permutations is a graph invariant{, opening up the possibility of lifting any graph invariant to a possibly more expressive invariant based on the correlation matrix}. We focus on the graph spectrum and show numerically that at depth $p=2$ it distinguishes $10{,}486$ pairs of regular graphs that are cospectral with respect to their adjacency matrices. We further derive a finite-shot sample-complexity bound for estimating the average ZZ eigenvalue gap.

\item We study IQP distributions in greater detail. Although the Hamming weight distribution alone fails to distinguish some non-isomorphic graphs, we introduce a finer invariant based on the multiplicities of individual probabilities within each Hamming weight sector. This invariant distinguishes all non-isomorphic graph pairs tested numerically.

\item Finally, we characterise the IQP distribution in terms of the underlying graph structure for particular parameter choices. At $\theta=\pi/2$, the output is determined by the parity of the vertex degrees, while at $\theta=\pi/4$ the circuit is Clifford and its distribution can be described explicitly through its connection to the Hidden Linear Function problem of Bravyi, Gosset and K\"onig~\cite{Bravyi_2018}.
\end{enumerate}

For the numerical experiments on regular graphs, we use graph instances generated by Meringer's regular graph generator~\cite{meringer1999fast}. In particular, these instances are used to construct the cospectral regular graph pairs studied in Section~\ref{sec: QAOA ZZ matrix}. The graph instances used in our numerical experiments are available in the accompanying repository~\cite{Beasley_Python_Graph_collections}, allowing independent reproduction of the numerical experiments using the same graph instances.

\section{Preliminaries and notation}

A graph $G$ is described in terms of its vertices $V$ and edges $E$, $G=(V,E)$, where $|V|=n$ is the size of the graph, and each edge $(i,j)\in E$ satisfies $i,j\in V$. A graph $G$ can also be described in terms of its adjacency matrix.

\begin{definition}{\textbf{Adjacency matrix}}\label{def: Adjacency matrix}
    \\The adjacency matrix of a graph, denoted $A$, is the matrix with elements
    \begin{align*}
        A_{ij}=\begin{cases}
            1 &  \text{if } (i,j) \in E, \\
            0 &  \text{if } (i,j) \notin E.
        \end{cases}
    \end{align*}
\end{definition}

Formally, a graph isomorphism is defined as follows.~\cite{Grohe_2020}

\begin{definition}{\textbf{Graph Isomorphism}}\label{def: Graph isomorphism}
    \\An isomorphism from a graph $G = (V, E)$ to a graph $G' = (V', E')$ is a bijection $\pi:V \to V'$ such that $(u,v) \in E \Longleftrightarrow(\pi(u),\pi(v))\in E'$ for every $u,v \in V$. We denote isomorphism between graphs $G$ and $G'$ as $G\cong G'$.
\end{definition}

When $V=V'=\{1,\dots,n\}$, an isomorphism may therefore be identified with a permutation $\pi\in S_n$.

\subsection{IQP circuits}
IQP circuits are one of the candidates for showing that some computational tasks are easy for a quantum computer but difficult or impossible for classical computers to compute efficiently. {One example is certain problems concerning partition functions of the Ising model~\cite{Bremner_2017}.} 
A general expression {for the family of} IQP circuits we are investigating is found in~\cite{Shepherd_2009}. It is defined as ,
\begin{definition}{\textbf{General IQP Circuit}}\label{def: general-IQP} 
    \[
\exp\left(\sum_{\textbf{p}}i\theta_\textbf{p}\bigotimes_{j:p_j=1}X_j\right)\ket{0^n}
    \]
\end{definition}
where X is the Pauli X gate, $X = \ket{0}\bra{1} + \ket{1}\bra{0}$ and $(\theta_{\textbf{p}},\textbf{p})\in\mathbb{R}\times\mathbb{F}_2^n$.

We examine a subset of these circuits whose structure reflects that of an n-vertex simple, connected graph $G=(V,E)$. 
To code the structure of the graph into the circuit, each $v\in V$ corresponds to a qubit and each $(i,j)\in E$ corresponds to applying the gate $\exp(i\theta X \otimes X)$ between qubits $i$ and $j$. A simple example showing how a graph can be translated into a quantum circuit is shown in Figure~\ref{fig: Graph-Circuit Equivalence}

\begin{definition}{\textbf{IQP Circuit That Encodes a Simple Connected Graph}}\label{def: IQP circuit} 
\\For a graph $G=(V,E)$ and parameter $\theta\in \mathbb{R}$, 
\[
\ket{\operatorname{IQP}_G}\coloneqq \prod_{(i,j) \in E}\exp\left(i\theta X_i X_j\right)\ket{0^n}.
\]
\end{definition}

\begin{figure}[tb!]
    \centering
    \makebox[\textwidth]{%
        \includegraphics[width=1\textwidth, height=0.34\textwidth]{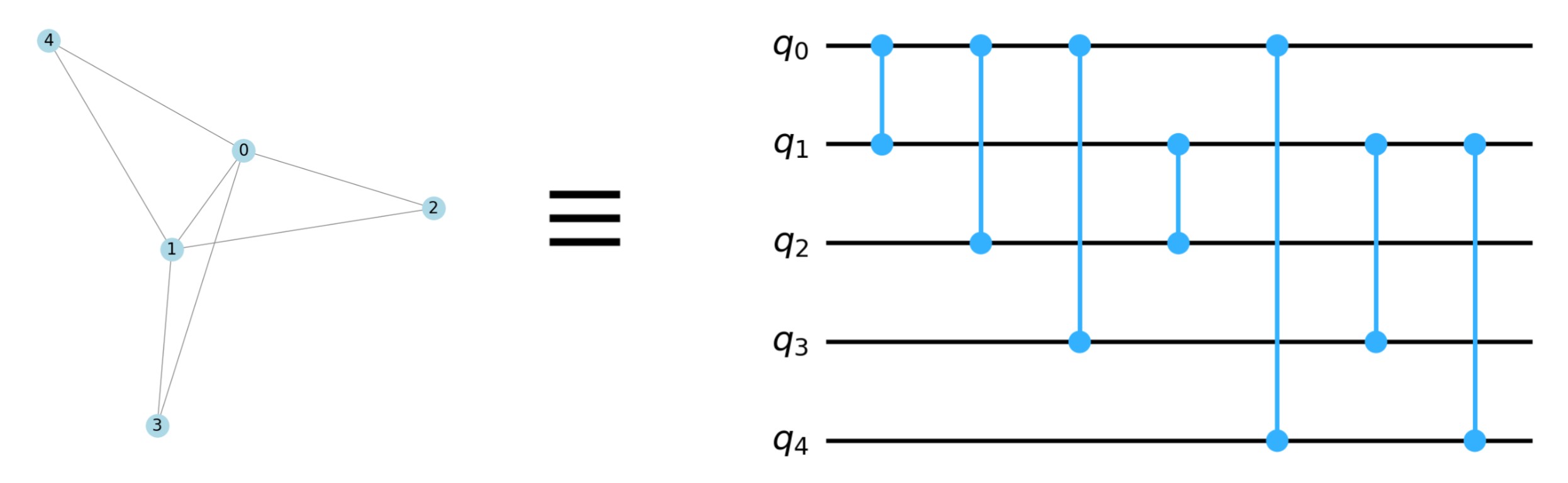}%
    }
    \caption{Simple example of equivalence between graphs and IQP circuits}
	\label{fig: Graph-Circuit Equivalence}
\end{figure}\vspace{2mm}

\subsection{QAOA circuits}

Quantum Approximate Optimisation Algorithm (QAOA) circuits, originally proposed by Farhi, Goldstone and Gutmann in 2014~\cite{farhi_2014}, are designed as a discretisation of the Quantum Adiabatic Algorithm. QAOA circuits are {often applied to} NP-hard optimisation problems such as Max-Cut, the Knapsack Problem and Job Scheduling. {Further, these circuits have been shown to be hard to simulate classically, under reasonable computational complexity assumptions}~\cite{farhi_2019}.

\begin{definition}{\textbf{General QAOA Circuit}}\label{def: General QAOA circuit}
\\The level p QAOA circuit with parameters $\gamma = (\gamma_0,\dots,\gamma_{p-1})$ and $\beta = (\beta_0,\dots,\beta_{p-1})$ is given by
\begin{align}
    \prod_{q=0}^{p-1}\left(\prod_{i =1}^ne^{-i\beta_{q}X_i}\prod_{\alpha=1}^me^{-i\gamma_{q}C_\alpha}\right)\cdot H^{\otimes n} \ket{0^n},
\end{align}
where $C(z) = \sum_{\alpha=1}^m C_\alpha(z).$ {and $C_{\alpha}(z)$ is a simple classical ``cost'' term that generally depends on at most a few input bits}. 
\end{definition}

In our case, we take inspiration from Mario Szegedy's {investigation} `What do QAOA energies reveal about graphs?'~\cite{szegedy_2019_qaoa} and explore these circuits in the context of solving the graph isomorphism problem.
Following Szegedy ~\cite{szegedy_2019_qaoa}, we associate a graph with a corresponding QAOA circuit by associating the cost term with each edge and the mixer term with each vertex.

\begin{definition}{\textbf{QAOA Circuit That Encodes a Simple Connected Graph}}\label{def: QAOA circuit}
\\For a graph $G=(V,E)$, the level p QAOA circuit with parameters $\gamma = (\gamma_0,\dots,\gamma_{p-1})$ and $\beta = (\beta_0,\dots,\beta_{p-1})$ is given by
\begin{align}
    \ket{\operatorname{QAOA}_G}\coloneqq \prod_{q=0}^{p-1}\left(\prod_{i\in V}e^{-i\beta_{q}X_i}\prod_{(j,k)\in E}e^{-i\gamma_{q}C_{(j,k)}}\right)\cdot H^{\otimes n} \ket{0^n},
\end{align}
where $C_{(j,k)} = \frac{1}{2}(I+Z_jZ_k).$
\end{definition}

In their original paper, Farhi, Goldstone and Gutmann show that, as the number of layers increases ($p\to\infty$), there exist parameters $\gamma, \beta$ for which QAOA approaches the optimum of the corresponding cost function~\cite{farhi_2014}. Their formulation considers maximising the cost function, whereas here we minimise the Min-Uncut cost
\[
C=\sum_{(j,k)\in E}C_{(j,k)}.
\]
These formulations are equivalent up to an additive constant and a sign change, since the Min-Uncut and Max-Cut cost Hamiltonians satisfy
\[
C_{\mathrm{MinUncut}}=|E|I-C_{\mathrm{MaxCut}}.
\]
Consequently, the highest-energy state in the Max-Cut formulation corresponds to the ground state of the Min-Uncut Hamiltonian used here. Thus, in the limit $p\to\infty$, there exist QAOA parameters for which the circuit output can approach the initially unknown ground-state subspace of $C$.

To calculate the cost $C$, we first run the circuit, then measure each qubit to obtain an output bitstring $x$. For each edge, $(j,k) \in E$, we get a contribution $x_j \odot x_k$ added to the cost, where $\odot$ is the XNOR operator (the NOT of the XOR). We sum these contributions {and average over many runs} to get an estimation of the {expected cost}. However, it is important to recognise that, as the size of the graph gets larger, {errors affecting the quantum circuit} can introduce inaccuracies~\cite{xue2019quantumnoise,Scriva2024shotnoise}.

\subsection{Distribution behaviour under permutations}

{In this section we state some useful results about the effect of graph permutations on QAOA and IQP circuits; see~\cite{shaydulin21} for a more detailed study of the effect of symmetries on QAOA output distributions.}

\begin{restatable}{lemma}{iqppermutationlemma}\label{lemma: IQP permuting nodes permutes bits}
\textbf{Permuting the vertices of a graph permutes the corresponding bits in the IQP circuit's distribution}
\\ Let $G'$ be the graph produced by applying permutation $\pi$ to the graph $G$. Let $U_\pi$ be the unitary operator that permutes qubits according to the permutation $\pi$ and let
    \[
    \prod_{(i,j) \in E}\exp\left(i\theta X_i X_j\right)\ket{0^n}=\ket{\text{IQP}_G}
    \]
    be the IQP circuit that corresponds to the graph $G$.
    Then the IQP circuit that corresponds to the graph $G'$ is 
    \[
    \ket{\text{IQP}_{G'}} = U_\pi \ket{\text{IQP}_G}.
    \]
    The IQP circuit's distribution that corresponds to $G'$ is the IQP circuit's distribution corresponding to $G$ but with its support permuted according to the permutation $\pi$.
\end{restatable}

The proof is deferred to Appendix~\ref{appendix: proof of IQP permuting nodes permutes bits}.

\begin{restatable}{lemma}{qaoapermutationlemma}\label{lemma: QAOA permuting nodes permutes bits}
    \textbf{Permuting the vertices of a graph permutes the corresponding bits in the QAOA circuit's distribution}
    \\ Let $G'$ be the graph produced by applying a permutation $\pi$ to the graph $G$. Let $U_\pi$ be the unitary operator that permutes qubits according to $\pi$, and let
    \[
    \prod_{q=0}^{p-1}\left(\prod_{i\in V}e^{-i\beta_{q}X_i}\prod_{(j,k)\in E}e^{-i\gamma_{q}C_{(j,k)}}\right)\cdot H^{\otimes n} \ket{0^n}
    = \ket{\text{QAOA}_G}
    \]
    be the level $p$ QAOA circuit corresponding to $G$.\vspace{2mm}\\
    Then the QAOA circuit corresponding to $G'$ is
    \[
    \ket{\text{QAOA}_{G'}} = U_\pi \ket{\text{QAOA}_G}.
    \]
    In particular, the QAOA distribution corresponding to $G'$ is that of $G$ under the permutation $\pi$.
\end{restatable}

The proof is deferred to Appendix~\ref{appendix: proof of QAOA permuting nodes permutes bits}.

\subsection{The Hamming weight distribution}

The \textbf{Hamming weight} of a bitstring $x$, denoted $|x|$, is defined as the number of non-zero bits in the string, $|x|=\sum_{i=1}^n x_i$. It is clear that the Hamming weight is invariant under permutations of the bits.


    

\begin{definition}{\textbf{The Hamming weight distribution}}\label{def: Hamming weight distribution}
\\Let $P:\{0,1\}^n \to [0,1]$ be a probability distribution, i.e.\ $\sum_{y \in \{0,1\}^n} P(y) = 1$. The Hamming weight distribution of $P$, denoted $P_{HW}:\{0,\dots,n\}\to[0,1]$, is defined by
\[
P_{HW}(k) = \sum_{y\in\{0,1\}^n:\,|y|=k}P(y), \qquad k \in \{0,\dots,n\}.
\]
\end{definition}
Since Hamming weight groups partition $\{0,1\}^n$,
\[
\sum_{k=0}^{n}P_{HW}(k) = \sum_{y\in\{0,1\}^n}P(y) = 1,
\]
so $P_{HW}$ is a valid probability distribution.

\begin{lemma}{\textbf{Isomorphic graphs produce identical Hamming weight distributions}} \label{lemma: HW dist invariance under perm}
\\ Let $G$ and $G'$ be isomorphic graphs related by some permutation of their vertices $\pi \in S_n$, and let $P_G$ and $P_{G'}$ denote their corresponding probability distributions, related by
\[
P_{G'}(y) = P_G(\pi(y)) \quad \text{for all } y \in \{0,1\}^n.
\]
Then for each $k \in \{0,\dots,n\}$,
\[
\sum_{y\in\{0,1\}^n:|y|=k}P_G(y) = \sum_{y\in\{0,1\}^n:|y|=k}P_{G'}(y).
\]
That is, $G$ and $G'$ produce identical Hamming weight distributions.
\end{lemma}

\begin{proof}
Fix $k \in \{0,\dots,n\}$. Then,
\begin{align*}
    \sum_{y\in\{0,1\}^n:|y|=k}P_{G'}(y) &= \sum_{y\in\{0,1\}^n:|y|=k}P_{G}(\pi(y))\\
    &= \sum_{y\in\{0,1\}^n:|y|=k}P_{G}(y),
\end{align*}
where the second equality follows from invariance of the Hamming weight under permutations,
which allows us to rewrite the sum over permuted bitstrings $\{\pi(y) : |y|=k\}$ as a sum over $\{y:|y|=k\}$, since $\pi$ permutes bitstrings of weight $k$ among themselves. Therefore, the Hamming weight distributions of $G$ and $G'$ are identical.
\end{proof}

Note that the implication in Lemma~\ref{lemma: HW dist invariance under perm} goes only one way. It is not the case that identical Hamming weight distributions imply graph isomorphism. However, it does mean that if we have graphs whose Hamming weight distributions are distinct, we know that they cannot be isomorphic to each other.


\section{Graph reconstruction using QAOA and IQP circuits}

IQP and QAOA circuits are very simple families of quantum circuits, so a basic question one might ask is whether the output probability distributions obtained from them are sufficient to determine graph isomorphism, even in principle, or if these distributions do not include enough information to do so. In this section we show that the original graph can be retrieved from the corresponding IQP and QAOA output distributions. In particular, this allows for graph isomorphism to be determined in principle using these circuits, though not necessarily efficiently.

First we consider IQP circuits, where the probability of receiving output $x$ is
\begin{equation}
    p_x = |\langle x|e^{i \theta \sum_{(i,j)\in E} X_i X_j} |0 \rangle|^2.
\end{equation}
We take $\theta$ small {and positive} and use the Taylor expansion of the exponential $e^{i\theta A}= I + i\theta A + R$, where $R$ satisfies $||R||=O(\theta^2)$ to obtain
\begin{equation}
    p_x = \left|\delta_{x0} + i \theta \langle x|\sum_{(i,j)\in E} X_i X_j|0 \rangle + \epsilon_x \right|^2 = \begin{cases}
    |1+\epsilon_{0^n}|^2 & \text{if } x=0^n\\
    |i \theta + \epsilon_x |^2 & \text{if } |x|=2, x_i=1, x_j=1, (i,j) \in E\\
    |\epsilon_x|^2 & \text{otherwise}
    \end{cases}
\end{equation}
for some value $\epsilon_x$ such that $|\epsilon_x| \le (\theta|E|)^2 \le (\theta n^2)^2$ (because $|E| \le \frac{n(n-1)}{2} \le n^2$). \\
By the reverse triangle inequality, $|i \theta + \epsilon_x | \ge \theta-|\epsilon_x|$. So the probability of seeing each edge in $E$ is lower-bounded by $(\theta - \theta^2 n^4)^2=\theta^2(1 - \theta n^4)^2$, and the probability of seeing an edge not in $E$ is upper-bounded by $\theta^4 n^8$. Taking $\theta \le 1/(2n^4)$, the former quantity is lower-bounded by $\theta^2/4$. After $k$ samples, the probability of seeing any non-existent edges is at most $k \theta^4 n^{10}$ by a union bound, while the probability of failing to see any edges in $E$ is at most
\begin{equation}
    |E| (1-\theta^2/4)^k \le n^2 e^{-\theta^2k/4}.\\
\end{equation}
For concreteness, take $\theta = 1/n^6$, $k=n^{13}$; then the probability of seeing any non-existent edges is at most $1/n$ while the probability of failing to see any edges in $E$ is exponentially small in $n$. $(\le n^2e^{-n/4})$.

Next we consider QAOA circuits, where we use a somewhat similar idea. Fix $p=1$. Then, up to an irrelevant global phase, the state produced by the QAOA circuit is
\begin{equation}
    e^{-i \beta \sum_j X_j}
    e^{-i \gamma/2 \sum_{(j,k)\in E} Z_j Z_k} \ket{+}^{\otimes n} = e^{-i \beta \sum_j X_j}
    \left(I- i \gamma/2 \sum_{(j,k) \in E} Z_j Z_k\right) \ket{+}^{\otimes n} + O(\gamma^2).
\end{equation}

We consider the correlation $\braket{Z_aZ_b}$. Since every computational-basis measurement yields values of $Z_aZ_b$ for all pairs $a,b$ simultaneously, we can estimate all these correlations from the same set of circuit samples. Writing $c=\cos(2\beta)$ and $s=\sin(2\beta)$ and using that $\bra{+}Z\ket{+}=\bra{+}Y\ket{+}=0$,
\begin{align*}
\braket{Z_aZ_b}&=
     \bra{+}^{\otimes n} \left(I+i \gamma/2 \sum_{(j,k) \in E} Z_j Z_k \right) e^{i \beta \sum_j X_j} Z_a Z_b e^{-i \beta \sum_j X_j}
    \left(I- i \gamma/2 \sum_{(j,k) \in E} Z_j Z_k\right) \ket{+}^{\otimes n} + O(\gamma^2)\\
    &= 
    \bra{+}^{\otimes n} \left(I+i \gamma/2 \sum_{(j,k) \in E} Z_j Z_k \right) 
     \left(c^2Z_aZ_b +cs\left(Z_aY_b+Y_aZ_b\right)+s^2Y_aY_b\right)
    \\[-2mm] &\qquad \left(I- i \gamma/2 \sum_{(j,k) \in E} Z_j Z_k\right) \ket{+}^{\otimes n} + O(\gamma^2)\\
    \\[-4mm]
    &= i \gamma/2 \bra{+}^{\otimes n} \left[\sum_{(j,k) \in E} Z_j Z_k,c^2Z_aZ_b +cs\left(Z_aY_b+Y_aZ_b\right)+s^2Y_aY_b\right]\ket{+}^{\otimes n} + O(\gamma^2)\\
     &=\begin{cases}
         i \gamma/2 \bra{+}^{\otimes n} cs\left(-2iX_a-2iX_b\right)\ket{+}^{\otimes n} + O(\gamma^2)& \text{if $(a,b)\in E$}\\
        O(\gamma^2) &  \text{if $(a,b)\notin E$}
     \end{cases}\\
    &= \begin{cases}
        \gamma\sin(4\beta) + O(\gamma^2)& \text{if $(a,b)\in E$}\\
        O(\gamma^2) &  \text{if $(a,b)\notin E$}
    \end{cases}
\end{align*}
Taking $\beta =\pi/8$, we therefore have
\[
\braket{Z_aZ_b} = \gamma A_{ab}+O(\gamma^2n^2).
\]
Then choosing $\gamma =  \Theta(n^{-4})$, the cases where $(a,b)\in E$ and $(a,b)\notin E$ can be separated by a constant fraction of $\gamma$, so it is sufficient to estimate $\braket{Z_a Z_b}$ up to $O(\gamma)$. Using a union bound over the edges, it is sufficient to take $\widetilde{O}(n^8)$ samples to learn all the edges in the graph.

%


\section{Distinguishing nonisomorphic graphs using QAOA}

{In Ref.~\cite{szegedy_2019_qaoa}, the `energy' of states produced by QAOA circuits for different graphs using identical angle sequences $\gamma,\beta \in [0,2\pi]^{p}$ is used to check graph isomorphism}. Each angle $\gamma_q,\beta_q \in \gamma,\beta$ is independent and is uniformly chosen from the interval $[0,2\pi]$.

\begin{definition}{\textbf{QAOA Energy}}\label{def: QAOA Energy}
  \\ If we denote the state in Definition~\ref{def: QAOA circuit} as $\ket{\gamma,\beta}$ and define $C = \sum_{(j,k)\in E} C_{(j,k)}$ then the energy is defined as
  \begin{align*}
    E(G,\gamma,\beta)\coloneqq \bra{\gamma,\beta}C\ket{\gamma,\beta}.
    \end{align*}
\end{definition}

Szegedy shows that the energy, for isomorphic graphs, is equal as long as the circuit is run with identical parameters.

To compare graphs, it is useful to define the average energy gap between graphs.

\begin{definition}{\textbf{Average Energy Gap}}\label{def: Energy gap}
    \\For two graphs $G_1 \text{ and } G_2,$ the average energy gap is defined as 
    \begin{align*}
        |\Delta E(G_1,G_2,p)|\coloneqq  \mathbb{E}|E(G_1,\gamma,\beta)-E(G_2,\gamma,\beta)|, \hspace{50pt} \gamma,\beta \text{ uniform in } [0,2\pi]^{p}.
    \end{align*}
    We will often denote this as $|\Delta E|$.
\end{definition}

The magnitude of the energy gap is important because it affects the number of shots required to resolve an energy difference experimentally.

Szegedy's work motivates the following conjecture.

\conjecture{\textbf{All connected graphs can be distinguished via their QAOA energy.}}\label{conj: Szegedy}
\\Let $G_1$ and $G_2$ be two non-isomorphic connected graphs. Then there exists some level $p>0$ and a choice of parameters $\gamma,\beta \in [0,2\pi]^p$ such that
\[
E(G_1,\gamma,\beta)\neq E(G_2,\gamma,\beta).
\]

Rather than restricting attention to the expectation value of the cost Hamiltonian, we investigate whether additional information in the QAOA output distribution can distinguish graphs at lower circuit depths. In particular, we consider the Hamming weight distribution introduced in Definition~\ref{def: Hamming weight distribution}, which is invariant under graph isomorphism by Lemma~\ref{lemma: HW dist invariance under perm}. This motivates the following conjecture.

\conjecture{\textbf{All connected graphs can be distinguished via the Hamming weight distribution produced by their corresponding QAOA circuit.}}\label{conj: QAOA Hamming weight distinguishability}
\\Let $G_1$ and $G_2$ be two non-isomorphic connected graphs, and let $P_G(k)$ denote the probability of measuring a bitstring with Hamming weight $k$ from the QAOA circuit \ref{def: QAOA circuit}, $k\in \{0,\dots,n\}$. Then there exists a depth $p>0$, a choice of parameters $\gamma,\beta \in [0,2\pi]^p$ and some $k\in\{0,\dots,n\}$ such that
\[
P_{G_1}(k)\neq P_{G_2}(k).
\]
Our numerical experiments have not produced a counterexample to either conjecture.

For either conjecture to yield a polynomial-time graph-isomorphism algorithm, it would additionally be necessary for the required circuit depth and number of measurements to scale polynomially with the graph size, and for suitable circuit parameters to be chosen efficiently.

Szegedy identified an important potential limitation of graph-isomorphism tests based solely on the QAOA energy. In particular, he found a sequence of non-isomorphic graph pairs for which the energy gap appears to decrease exponentially {for fixed $p$} as the graph size increases~\cite{szegedy_2019_qaoa}. In such cases, resolving the energy difference may require an impractically large number of measurements. Moreover, for some non-isomorphic graph pairs, the numerically computed gap can be exactly zero in floating-point arithmetic at the circuit depths considered.

We therefore investigate whether the Hamming weight distribution retains useful distinguishing information for such hard-to-distinguish graph pairs.

To compare Hamming weight distributions, we first need to introduce the total variation distance (TVD)
\begin{definition}{Hamming Weight Total Variation Distance (TVD)}\label{def: HW TVD}
\\ For Hamming weight distributions $P(k)$ and $Q(k)$ on $\{0,\dots,n\}$, 
    \[
    \operatorname{TVD}(P,Q) = \frac{1}{2}\sum_{k =0}^n|P(k)-Q(k)|
    \]
\end{definition}
It takes values between 0 and 1. A TVD of 0 means the distributions are identical and a TVD of 1 means they have disjoint support.

We {will see} that the computed TVD between Hamming weight distributions is frequently non-zero even when the computed value of $|\Delta E|$ is exactly zero. Thus, information contained in the QAOA distribution can distinguish some graph pairs at circuit depths for which the energy alone cannot.

However, this improvement is not universal. There exist graph families for which both $|\Delta E|$ and the TVD remain extremely small across all tested depths, implying that a large number of samples is still required. In such cases, the limitation appears to be intrinsic to the similarity of the underlying QAOA output distributions.

In the following subsections, we examine representative examples of such hard-to-distinguish graph pairs and compare the performance of the two metrics averaged across 500 random parameter sets $\gamma, \beta \in [0,2\pi]^p$. 

All classical simulations were performed using double-precision floating-point arithmetic. No additional numerical tolerance or threshold was applied when evaluating the energy gap, Hamming weight TVD, or ZZ eigenvalue gap. Accordingly, when a numerically computed quantity is described as zero, this means that its computed floating-point value was exactly zero; this should not be interpreted as an analytic proof that the corresponding exact quantity vanishes.

The following plots include error bars indicating the standard error of the mean (SEM), which measures how far the sample mean is likely to be from the population mean. It is defined as $SEM = \hat{\sigma}/\sqrt{N}$, where $N$ is the number of samples and $\hat{\sigma}$ is the sample standard deviation.

\subsection{Miyazaki graphs}
In 1997, Takunari Miyazaki constructed families of coloured, connected 3-regular graphs for which McKay's canonical-labelling algorithm, as implemented in \textsc{nauty}, exhibits exponential behaviour~\cite{miyazaki1997complexity}. In Szegedy's paper, he devotes particular attention to 3-regular graphs and studies a pair of hard-to-distinguish graphs that are inspired by Miyazaki's graphs, shown in Figure~\ref{fig: Miyazaki graphs}. 

\begin{figure}[tb!]
    \centering
    \includegraphics[width=0.9\textwidth]{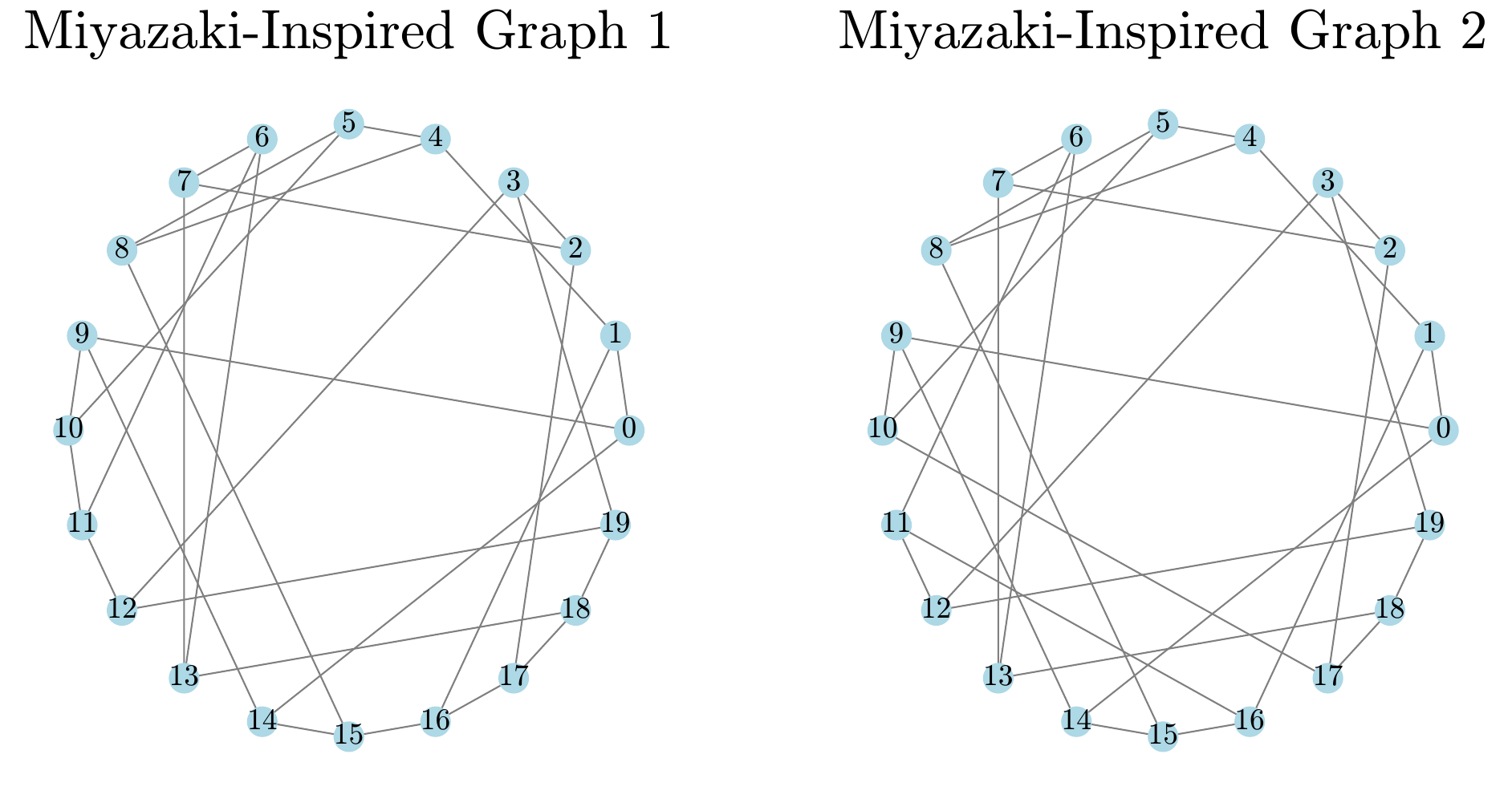}
    \caption{Two Miyazaki-inspired graphs.}
    \label{fig: Miyazaki graphs}
\end{figure}

Figure~\ref{fig: HW TVD vs Energy gap for Miyazaki graphs} shows that the numerically computed value of $|\Delta E|$ is zero for $p<3$ and becomes non-zero from depth $p=3$, consistent with Szegedy's results. In contrast, the computed Hamming weight TVD is non-zero for all tested depths. 
However, for this graph pair, both the TVD and $|\Delta E|$ remain {low} across all tested depths, {implying} that resolving these differences may take impractically many shots.

\begin{figure}[tb!]
    \centering
    \includegraphics[width=\textwidth]{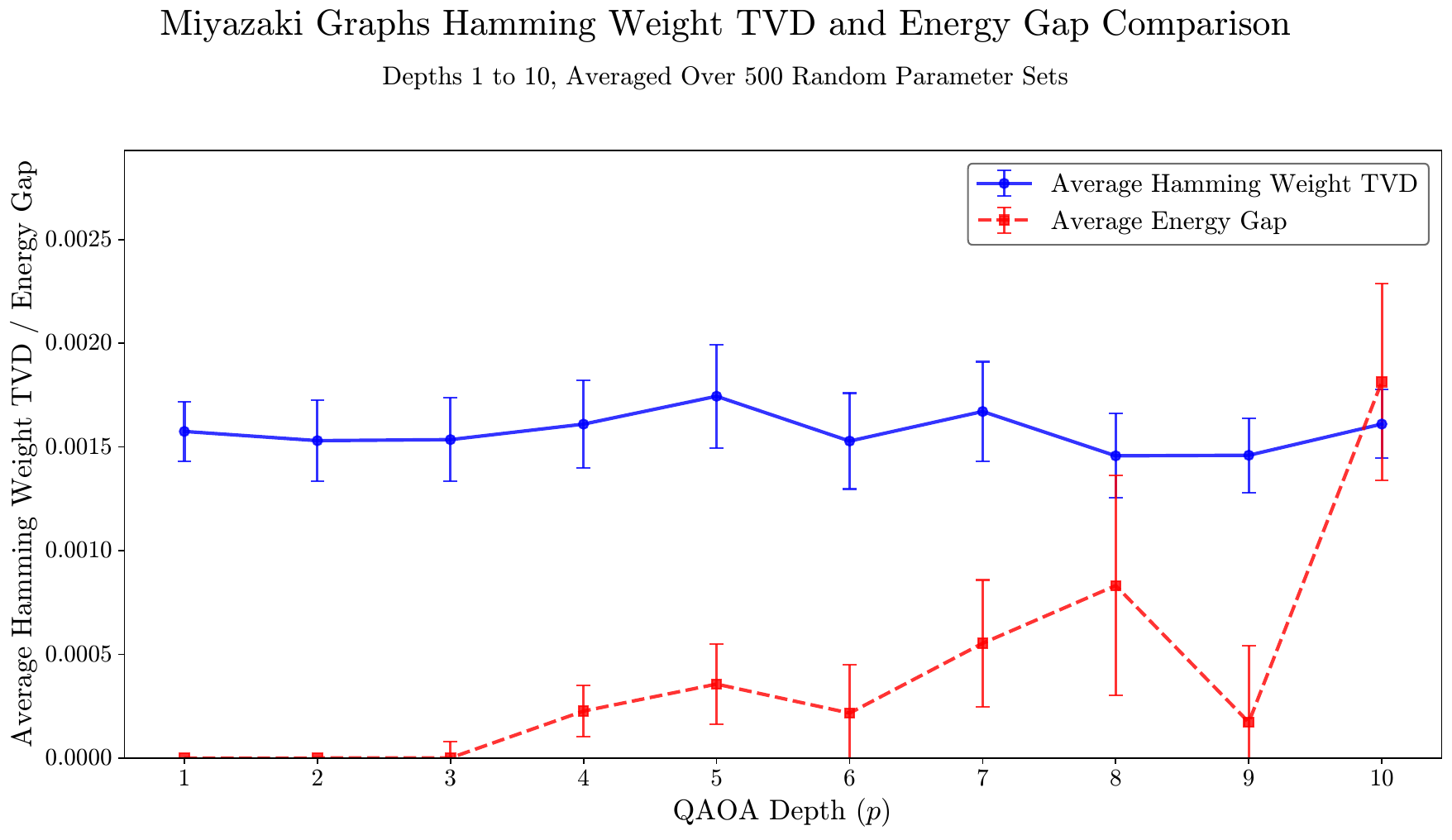}
    \caption{Comparison of Hamming weight TVD and energy gap for both Miyazaki-inspired graphs.}
    \label{fig: HW TVD vs Energy gap for Miyazaki graphs}
\end{figure}

\subsection{Circular and Möbius Ladder Graphs}

We next consider the circular and Möbius ladder graphs, shown in Figure~\ref{fig: CL and ML graph plots}. We denote these graphs by CL($m$) and ML($m$), respectively, where both graphs have $n=2m$ vertices. 

For $\operatorname{CL}(8)$ and $\operatorname{ML}(8)$, Figure~\ref{fig: HW TVD vs Energy gap for CL,ML(8)} shows that the numerically computed average QAOA energy gap is exactly zero for depths $p<4$ and becomes non-zero from $p=4$ onwards. In contrast, the computed Hamming weight TVD is non-zero at every tested depth, including $p=1$.

The same effect becomes more pronounced for the larger pair $\operatorname{CL}(10)$ and $\operatorname{ML}(10)$. As shown in Figure~\ref{fig: HW TVD vs Energy gap for CL,ML(10)}, the numerically computed energy gap is exactly zero for depths $p<5$, while the computed Hamming weight TVD is non-zero at every tested depth.

These examples reveal a structural limitation of the QAOA energy at low depth. As the graph size increases from $16$ to $20$ vertices, the minimum depth at which the computed energy gap becomes non-zero increases from $p=4$ to $p=5$, whereas the Hamming weight distribution distinguishes both pairs already at $p=1$. This demonstrates that distinguishing information can be present in the QAOA output distribution even when no difference in the expectation value is resolved numerically.

Despite this low-depth advantage, the Hamming weight distribution does not uniformly dominate the QAOA energy as a distinguishing statistic. In our broader numerical experiments, the QAOA energy already distinguishes most graph pairs effectively. The Hamming weight distribution is therefore most useful as a complementary statistic for cases in which the numerically computed energy gap is zero or becomes very small at low circuit depths.

\begin{figure}[tb!]
    \centering
    \includegraphics[width=0.9\textwidth]{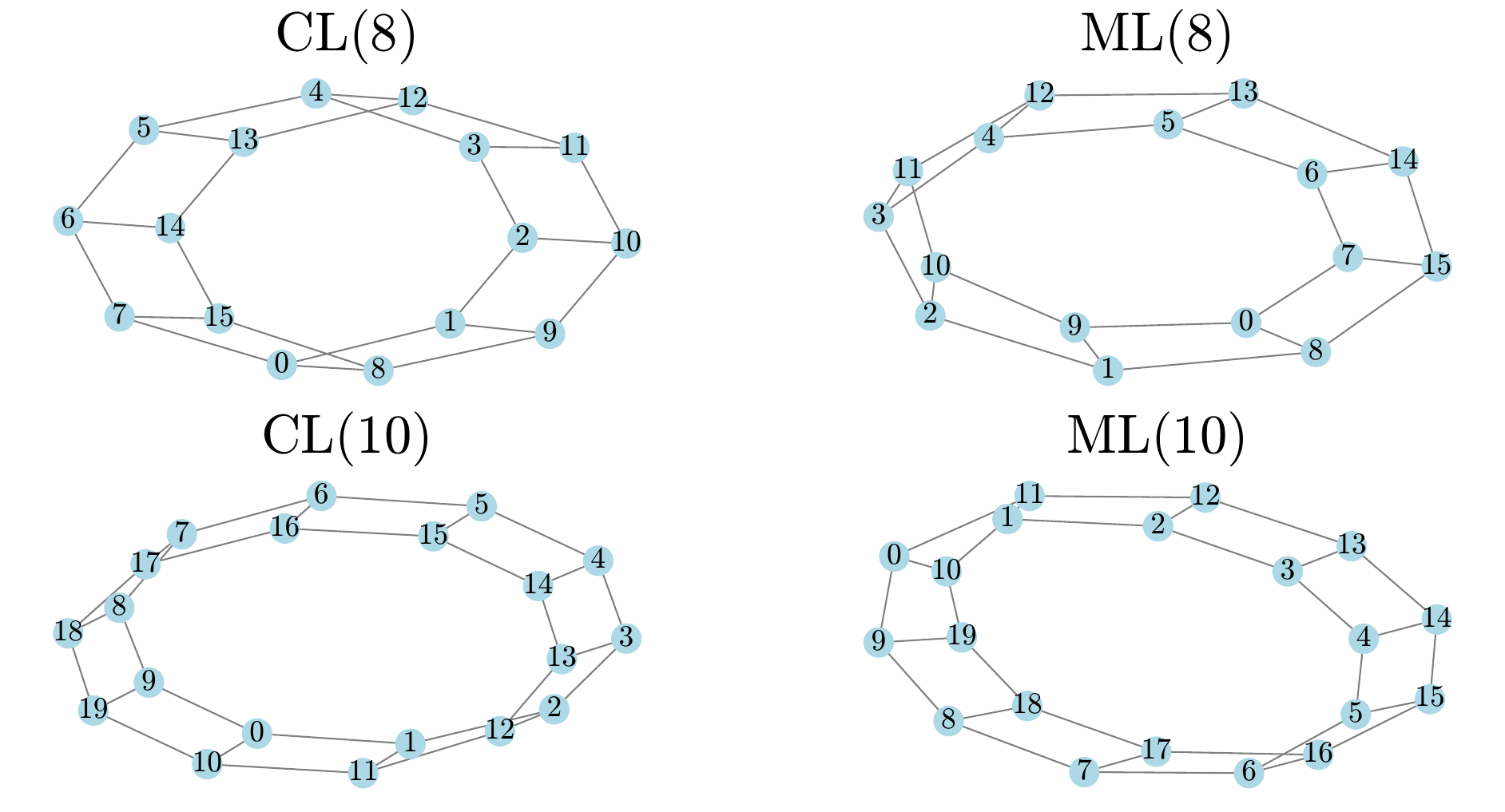}
    \caption{Circular and Möbius ladder graphs with 16 and 20 vertices.}
    \label{fig: CL and ML graph plots}
\end{figure}

\begin{figure}[tb!]
    \centering
    \includegraphics[width=\textwidth]{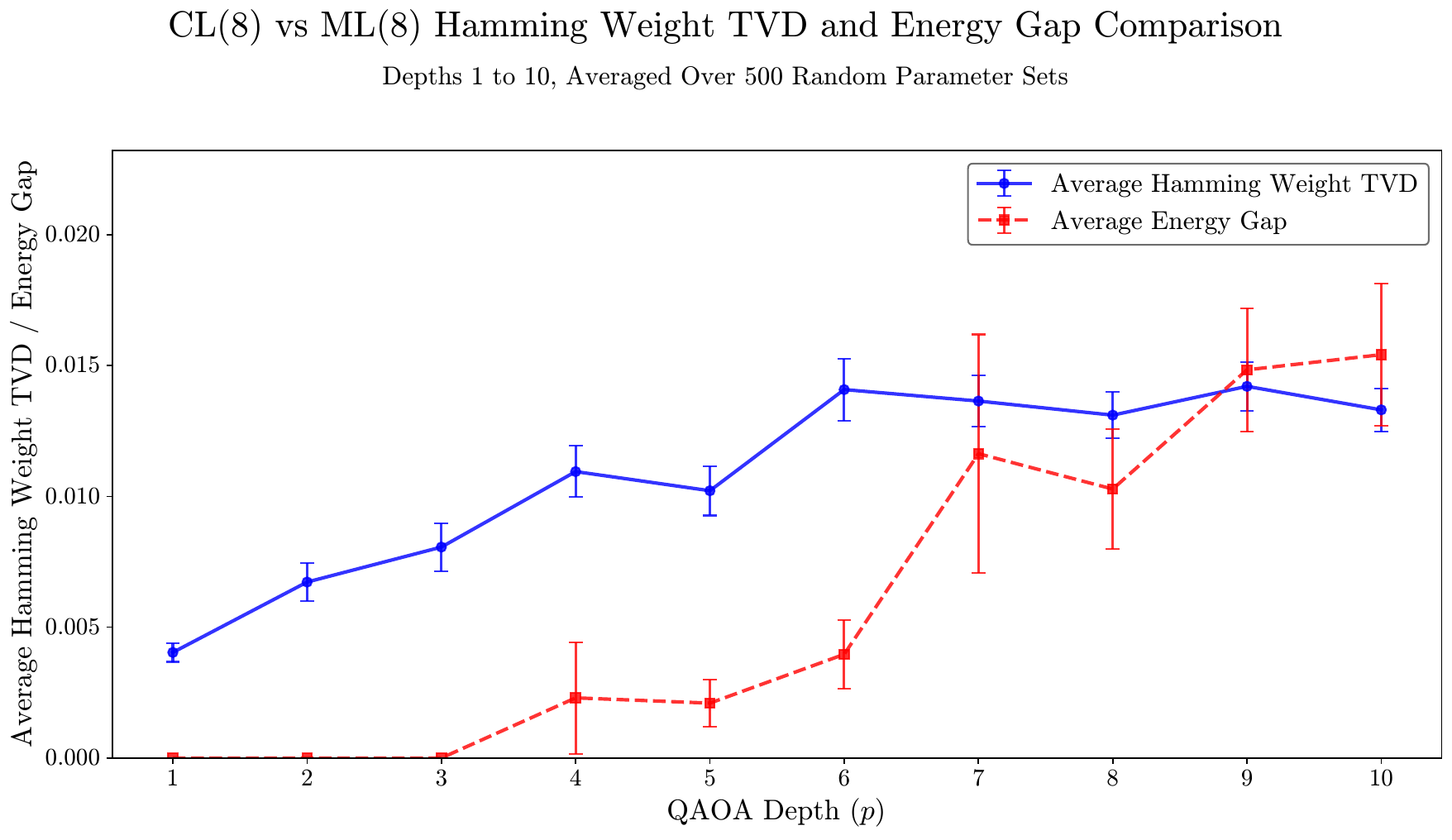}
    \caption{Comparison of Hamming weight TVD and energy gap for CL(8) and ML(8).}
    \label{fig: HW TVD vs Energy gap for CL,ML(8)}
\end{figure}

\begin{figure}[tb!]
    \centering
    \includegraphics[width=\textwidth]{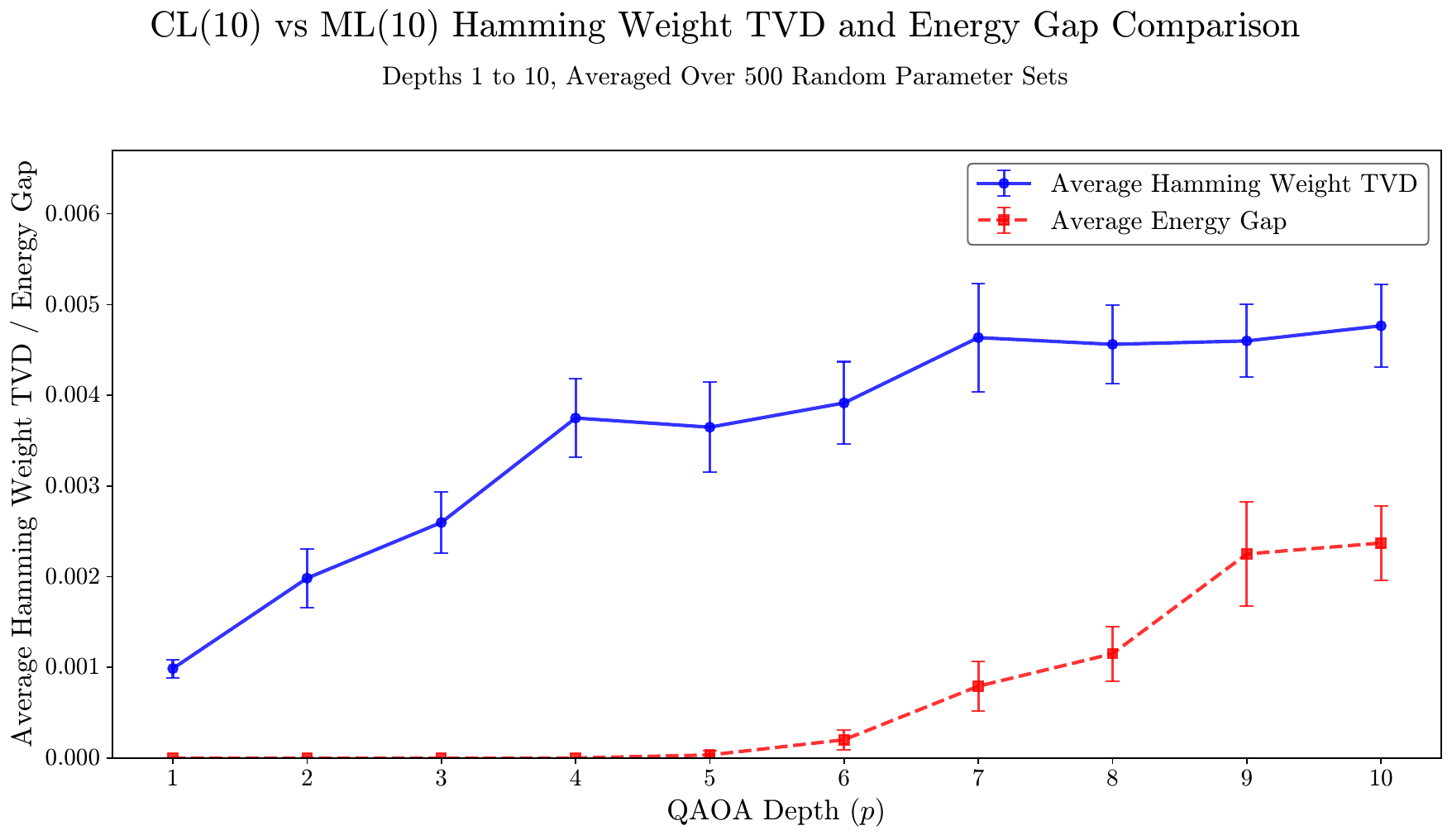}
    \caption{Comparison of Hamming weight TVD and energy gap for CL(10) and ML(10).}
    \label{fig: HW TVD vs Energy gap for CL,ML(10)}
\end{figure}


\section{Testing graph isomorphism using the matrix of correlators in QAOA}\label{sec: QAOA ZZ matrix}

In this section, we describe a general method for using QAOA to ``lift'' arbitrary  invariants under permutation similarity, based on the concept that even if the invariant does not distinguish the original graphs, in some cases it can be used to distinguish their correlation matrices.

\begin{definition}{\textbf{ZZ Correlation Matrix}}\label{def: ZZ correlation matrix}
\\For a graph $G$ whose corresponding probability distribution $P_G$ is produced by a level-$p$ QAOA circuit, its ZZ correlation matrix, denoted $M_G$, is defined by
\[
M_G\coloneqq \left({\braket{Z_iZ_j}_G}\right)_{i,j=1}^n
\]
where
\[
    \\{\braket{Z_iZ_j}_G} \coloneqq \sum_{x\in \{0,1\}^n} P_G(x)(-1)^{x_i\oplus x_j}.
\]
\end{definition}

We say that a function $f$ is  \emph{invariant under permutation similarity} if $f(M) = f(P M P^T)$ for any permutation matrix $P$.

\begin{lemma}{\textbf{Invariants under permutation similarity of the ZZ correlation matrix are preserved under permutations of the graph}}\label{thm: ZZ spectrum invariant under graph isomorphism}
    \\Let $G'$ be the graph obtained by applying a permutation $\pi\in S_n$ to the vertices of $G$, let $f$ be  invariant under permutation similarity, and suppose that the corresponding QAOA circuits are run with identical depth and parameters. Then 
    \[ 
    f(M_{G'}) = f(M_{G}).
    \]
\end{lemma}

\begin{proof} 
Let $U_\pi$ denote the unitary operator that permutes the qubits according to $\pi$. By Lemma~\ref{lemma: QAOA permuting nodes permutes bits}, 
\[ 
\ket{\mathrm{QAOA}_{G'}} = U_\pi\ket{\mathrm{QAOA}_{G}}, 
\] 
and 
\[ 
U_\pi Z_i U_\pi^\dagger = Z_{\pi(i)}. 
\] 

Therefore, for every $i,j\in\{1,\dots,n\}$, 

\begin{align*} \braket{Z_{\pi(i)}Z_{\pi(j)}}_{G'} &= \bra{\mathrm{QAOA}_{G'}} Z_{\pi(i)}Z_{\pi(j)} \ket{\mathrm{QAOA}_{G'}}
\\ &= \bra{\mathrm{QAOA}_{G}} U_\pi^\dagger Z_{\pi(i)}Z_{\pi(j)} U_\pi \ket{\mathrm{QAOA}_{G}}
\\ &= \bra{\mathrm{QAOA}_{G}} Z_iZ_j \ket{\mathrm{QAOA}_{G}}
\\ &= \braket{Z_iZ_j}_G. 
\end{align*} 

Hence, applying $\pi$ to the vertices of $G$ applies the same permutation to the rows and columns of $M_G$. 
Let $\Pi$ be the $n\times n$ permutation matrix corresponding to $\pi$. Then 
\[ 
M_{G'} = \Pi M_G\Pi^T. 
\]
Since $f$ is invariant under permutation similarity, $f(M_{G'}) = f(M_G)$.
\end{proof} 

\begin{corollary}\label{cor: distinct ZZ spectra imply nonisomorphism}
{If two graphs $G$ and $G'$ produce ZZ correlation matrices $M_G$, $M_{G'}$ for the same QAOA depth and parameters such that $f(M_{G'}) \neq f(M_{G})$ for some function $f$ which is invariant under permutation similarity,} then
\[
    G\not\cong G'.
\]
\end{corollary}

{Thus we see that correlation matrices provide a richer, QAOA-derived matrix representation of the graph, which can be thought of as an upgraded adjacency matrix and may allow nonisomorphism to be detected even when applying the invariant to the adjacency matrix fails. We remark that these matrices are connected to the higher-order density matrices studied in~\cite{szegedy_2019_qaoa} (Mario Szegedy, personal communication).

We will focus on the specific case of the spectrum of a matrix, which is invariant under permutation similarity. The spectrum is a well-known graph invariant that can often be used to certify nonisomorphism; however, many cases of cospectral nonisomorphic graphs are known.
Here we consider some examples of cospectral, regular graphs and show that the spectrum of the ZZ correlation matrix can allow us to distinguish between graphs that otherwise have indistinguishable spectra. 
For these numerical experiments, spectra were treated as numerically indistinguishable only when the corresponding computed eigenvalue differences were zero in floating-point arithmetic.

First, we tested depth-1 QAOA circuits (Figure~\ref{fig: Proportion of cospectral graphs that are indistinguishable at depth 1}), where we find that the majority of cospectral graph pairs that we tried were distinguishable using correlation matrices, although the fraction that were indistinguishable increased with problem size. In contrast, at depth $p=2$, every cospectral pair tested from the same set as in Figure~\ref{fig: Proportion of cospectral graphs that are indistinguishable at depth 1} was distinguished by the ZZ correlation spectrum.

\begin{figure}[tb!]
    \centering
    \includegraphics[width=\textwidth]{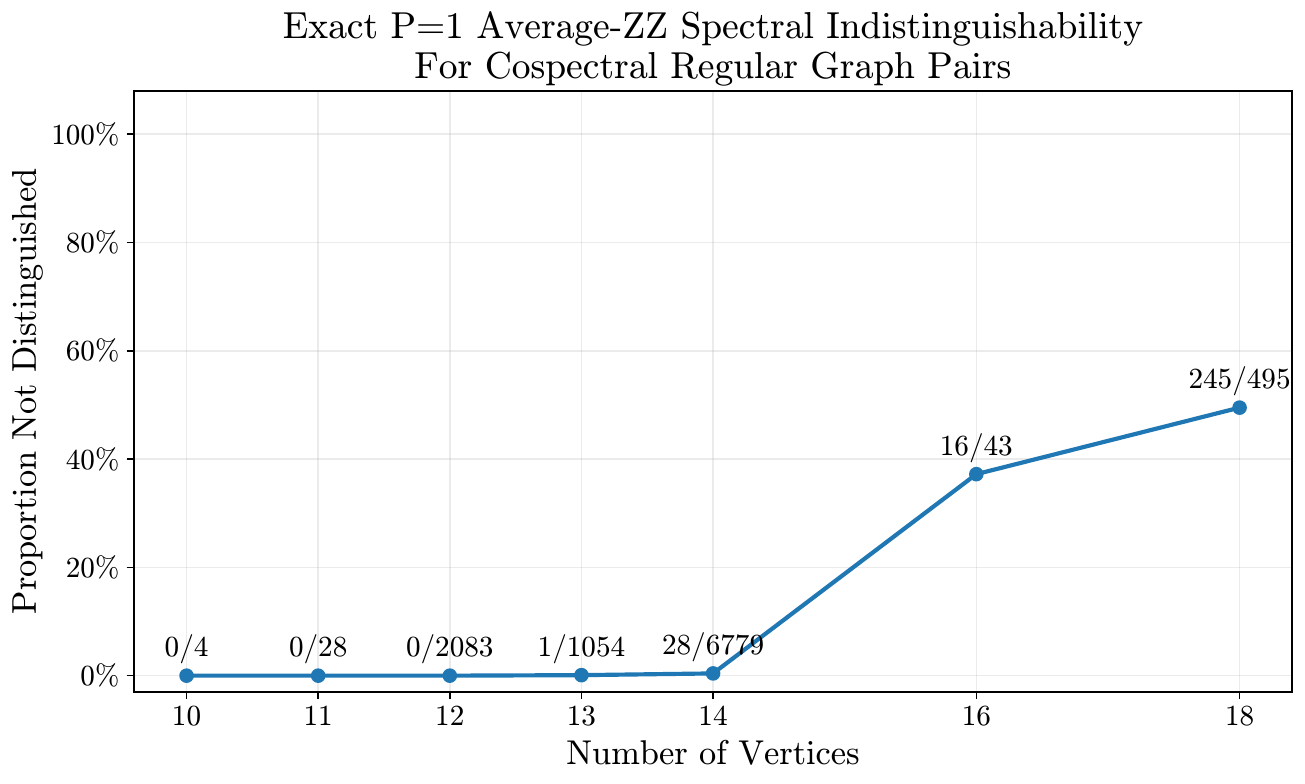}
    \caption{Proportion of cospectral graphs that are indistinguishable at depth 1.
    \\These graph pairs produced numerically indistinguishable ZZ correlation spectra for all 50 randomly chosen QAOA parameter sequences \(\gamma,\beta\in[0,2\pi]\).}
    \label{fig: Proportion of cospectral graphs that are indistinguishable at depth 1}
\end{figure}

In practice, it would not be possible to distinguish two nonisomorphic graphs from their correlation spectra if they were very close. We formalise this via a notion of average eigenvalue gap:
}

\begin{definition}{\textbf{Average ZZ eigenvalue gap}}\label{def: average ZZ eigenvalue gap}
\\ Let
\[
    \lambda_1\leq\lambda_2\leq\cdots\leq\lambda_n
\]
and
\[
    \mu_1\leq\mu_2\leq\cdots\leq\mu_n
\]
be the ordered eigenvalues of $M_G$ and $M_{G'}$, respectively. We define the average ZZ eigenvalue gap between $G$ and $G'$ to be
\[
    \Delta_{ZZ}(G,G') \coloneqq \frac{1}{n}\sum_{i=1}^n|\lambda_i-\mu_i|.
\]
\end{definition}

By Lemma~\ref{thm: ZZ spectrum invariant under graph isomorphism},
\[
    G \cong G' \quad \Longrightarrow \quad \Delta_{ZZ}(G,G')=0.
\]
Consequently, $\Delta_{ZZ}(G,G')>0$ is sufficient to certify that $G$ and $G'$ are non-isomorphic.

We next investigate whether the average ZZ eigenvalue gap can be reliably distinguished when the correlation matrix is estimated from a finite number of measurement shots. First, we prove theoretical bounds on how this quantity behaves. We will need the following theorem, which is based on the interpretation of the correlation matrix as a second moment matrix:

\begin{theorem}[Vershynin, Remark 5.6.5~\cite{Vershynin_2026}]
\label{thm: Vershynin covariance}
    Let $X$ be a random vector in $\mathbb{R}^n$, $n \ge 2$, and assume that for some $K\ge 1$,
    \[ \|X\|_2 \le K(\mathbb{E} \|X\|_2^2)^{1/2} \]
    almost surely. Write $\Sigma = \mathbb{E} [XX^T]$ and
    \[ \Sigma_m = \frac{1}{m} \sum_{i=1}^m X_i X_i^T \]
    where $X_i$ are independent random samples. Define the effective rank
    \[
    r(\Sigma)\coloneqq \frac{\operatorname{tr}(\Sigma)}{\|\Sigma\|}
    \]
    
    Then, for every positive integer $m$ and every $u\geq0$,
    \[ \| \Sigma_m - \Sigma \| \le C \left( \sqrt{\frac{K^2 r(\Sigma) (\log n+u)}{m}}+ \frac{K^2 r(\Sigma) (\log n+u)}{m} \right)\|\Sigma\| \]
    with probability at least $1-2e^{-u}$.
\end{theorem}

For empirical ZZ correlation matrices $\widetilde{M}_G$ and $\widetilde{M}_{G'}$, define the empirical average ZZ eigenvalue gap by 
\[
\widetilde{\Delta}_{ZZ}(G,G')\coloneqq \frac{1}{n}\sum_{i=1}^n \left|\lambda_i(\widetilde{M}_G) - \lambda_i(\widetilde{M}_{G'})\right|,
\]
where the eigenvalues are ordered so that $\lambda_1\leq\lambda_2\leq\cdots\leq\lambda_n$.

\begin{theorem}\label{thrm: ZZ shot scaling for non-iso graphs}
Let $G$ and $G'$ be graphs on $n\geq 2$ vertices and let ${M}_G$ and ${M}_{G'}$ be their exact ZZ correlation matrices obtained using the same QAOA depth and parameters. Assume that $\Delta_{ZZ}(G,G') = \delta>0$.
Let $\widetilde{M}_G$ and $\widetilde{M}_{G'}$ be obtained by averaging the single-shot correlator matrices from $m$ independent measurements of each circuit. Then 
\[
m = O\left(\frac{n^2 \log n}{\delta^2}\right)
\] 
is sufficient to achieve 
\[
\Pr\left[\left|\widetilde{\Delta}_{ZZ}(G,G')-\delta\right| \ge \delta/3\right] \le 0.01
\]
\end{theorem}

\begin{proof}
    Write $\widetilde{M}_G = M_G + E_G$ and $\widetilde{M}_{G'} = M_{G'} + E_{G'}$ where $E_G$ and $E_{G'}$ are the statistical error terms.
    Then, by Weyl's inequality,
    \[ 
    \left|\lambda_i(\widetilde{M}_G) - \lambda_i(M_G)\right| \le \|E_G\|, 
    \]
    and similarly for $G'$. Thus, by the triangle and reverse triangle inequalities,
    \begin{align*}
        \left|\widetilde{\Delta}_{ZZ}(G,G') -\delta\right|
        &\leq \frac{1}{n}\sum_{i=1}^n \left||\lambda_i(\widetilde{M}_G) - \lambda_i(\widetilde{M}_{G'})|-\left|\lambda_i({M}_G) - \lambda_i({M}_{G'})\right|\right|
        \\&\le\frac{1}{n}\sum_{i=1}^n \left|\lambda_i(\widetilde{M}_G) - \lambda_i(\widetilde{M}_{G'})-\lambda_i({M}_G) + \lambda_i({M}_{G'})\right|
        \\&\le\frac{1}{n}\sum_{i=1}^n \left(\left|\lambda_i(\widetilde{M}_G) - \lambda_i({M}_G)\right|+\left|\lambda_i(\widetilde{M}_{G'}) - \lambda_i({M}_{G'})\right|\right)
        \\&\le \|E_G\| + \|E_{G'}\| 
    \end{align*}
To bound $\|E_G\|$, we can appeal to Theorem~\ref{thm: Vershynin covariance}, using that each sample of correlations is generated by sampling a vector $X_i \in \{\pm1\}^n$ from some distribution, and computing $X_i X_i^T$. Thus we can take $K=1$ and use the fact that every diagonal entry of $M_G$ is one to obtain 
\[
r(M_G)=\frac{n}{\|M_G\|}.
\] 
Furthermore, 
\begin{align*}
    \|E_G\|&\leq C\left(\sqrt{\frac{r(M_G)(\log n+u)}{m}}+\frac{r(M_G)(\log n+u)}{m}\right)\|M_G\|
    \\&=C\left(\sqrt{\frac{n\|M_G\|(\log n+u)}{m}}+\frac{n(\log n+u)}{m}\right)
\end{align*}
Since every entry of $M_G$ is a correlation of two $\{\pm1\}$-valued variables, $|(M_G)_{i,j}|\leq 1$. Consequently,  
\[
\|M_G\|_F^2= \sum_{i,j}|(M_G)_{i,j}|^2\leq n^2,
\] 
so 
\[
\|M_G\|\leq \|M_G\|_F\leq n.
\]
Thus
\[
\|E_G\|\leq C\left(\sqrt{\frac{n^2(\log n+u)}{m}}+\frac{n(\log n+u)}{m}\right),
\]
except with probability at most $2e^{-u}$. The same bound holds for $E_{G'}$.
    
We choose $m\geq \log n+u$ so that the first term dominates. Setting $u=\ln(400)$, each matrix estimate fails with probability at most $0.005$. Hence, by a union bound, the bounds for $E_G$ and $E_{G'}$ hold simultaneously with probability at least $0.99$. Therefore, for some constant $C'$ and sufficiently large $n$,
\[
\left|\widetilde{\Delta}_{ZZ}(G,G')-\delta\right|
\leq
C'\sqrt{\frac{n^2\log n}{m}}
\]
with probability at least $0.99$.
Therefore, to achieve
\[
\left|\widetilde{\Delta}_{ZZ}(G,G')-\delta\right|
\leq \frac{\delta}{3},
\]
it is sufficient to take
\[
m=O\left(\frac{n^2\log n}{\delta^2}\right).
\]
\end{proof}


\begin{corollary}\label{cor: finite shot rejection of isomorphism}
Let $G$ and $G'$ be isomorphic graphs on $n\geq2$ vertices, and suppose that their QAOA circuits are run with the same depth and parameters. Let $\widetilde{M}_G$ and $\widetilde{M}_{G'}$ be obtained from $m$ independent measurements of each circuit. Then, for any threshold $\epsilon>0$,
\[
m=O\left(\frac{n^2\log n}{\epsilon^2}\right)
\]
is sufficient to ensure
\[
\Pr\left[\widetilde{\Delta}_{ZZ}(G,G')\geq \epsilon\right]\leq0.01.
\]
Hence, after this number of measurements, observing an empirical average ZZ eigenvalue gap of at least $\epsilon$ rejects the hypothesis that $G\cong G'$ with false-positive probability at most $0.01$.
\end{corollary}

\begin{proof}
If $G\cong G'$, then by Lemma~\ref{thm: ZZ spectrum invariant under graph isomorphism},
\[
\Delta_{ZZ}(G,G')=0.
\]
The same concentration argument used in the preceding theorem gives, with probability at least $0.99$,
\[
\widetilde{\Delta}_{ZZ}(G,G')
\leq
C'\sqrt{\frac{n^2\log n}{m}}
\]
for some constant $C'$. Choosing the implicit constant in
\[
m=O\left(\frac{n^2\log n}{\epsilon^2}\right)
\]
sufficiently large makes the right-hand side strictly smaller than $\epsilon$. Therefore
\[
\Pr\left[\widetilde{\Delta}_{ZZ}(G,G')\geq \epsilon\right]\leq0.01.
\]
\end{proof}

This corollary complements the preceding theorem. Theorem~\ref{thrm: ZZ shot scaling for non-iso graphs} bounds the error in estimating a non-zero ZZ eigenvalue gap, while Corollary~\ref{cor: finite shot rejection of isomorphism} bounds the probability that finite sampling produces an apparent gap of size at least $\epsilon$ when the graphs are in fact isomorphic.

In practice, these theoretical bounds may be somewhat pessimistic. We illustrate this in Figures~\ref{fig: Shot convergence of the average ZZ eigenvalue gap for CL(8) vs ML(8)} and~\ref{fig: Shot convergence of the average ZZ eigenvalue gap for G_{3770} vs G_{3774}}, which show the convergence of the sampled average ZZ eigenvalue gap to its classically simulated value as the number of shots increases. The dotted line denotes the gap obtained from the full probability distribution, while the solid line denotes the mean gap estimated from finite-shot samples.

\begin{figure}[tb!]
    \centering
    \includegraphics[width=\textwidth]{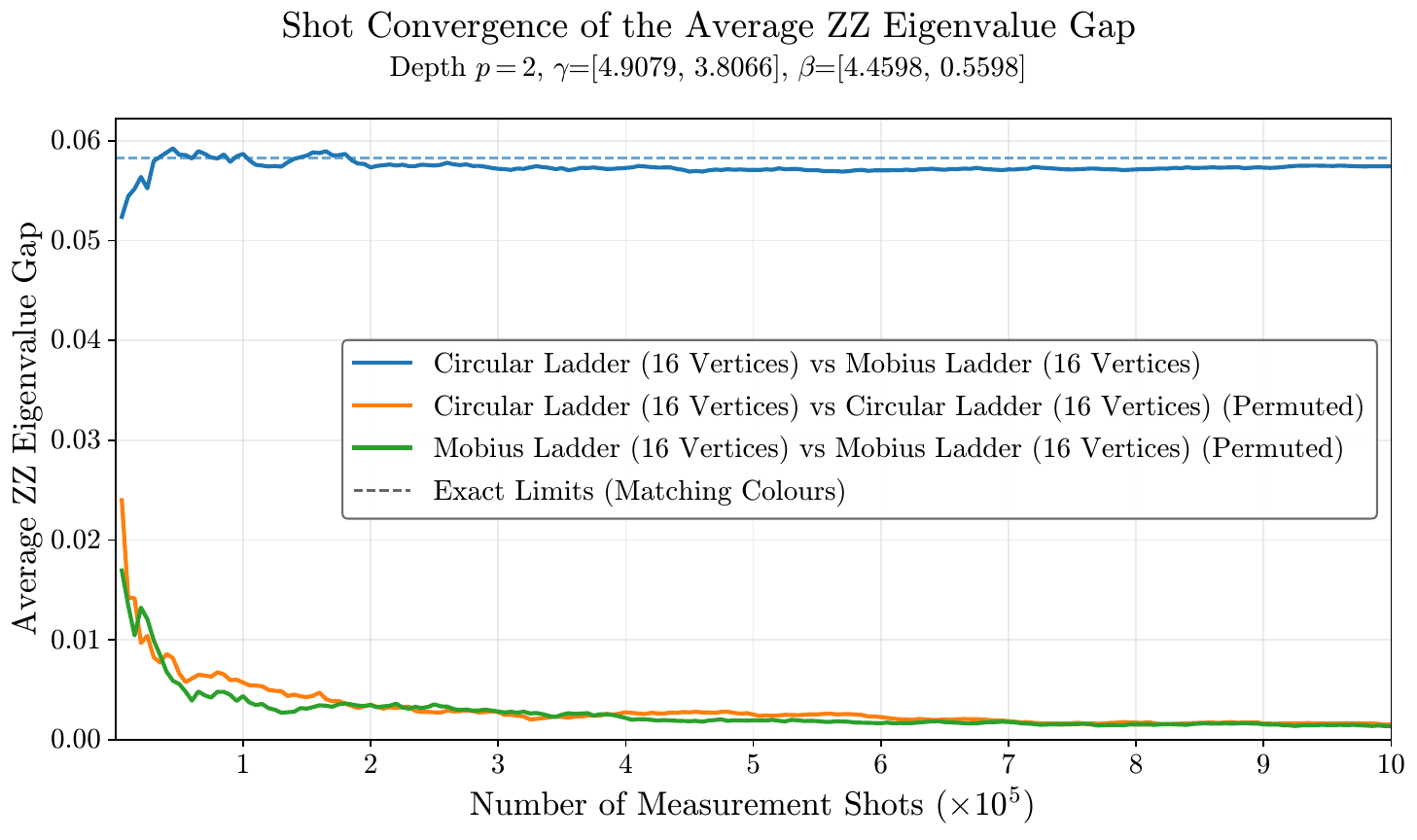}
    \caption{Shot convergence of $\Delta_{ZZ}$ for CL(8) vs ML(8).}
    \label{fig: Shot convergence of the average ZZ eigenvalue gap for CL(8) vs ML(8)}
\end{figure}

\begin{figure}[tb!]
    \centering
    \includegraphics[width=\textwidth]{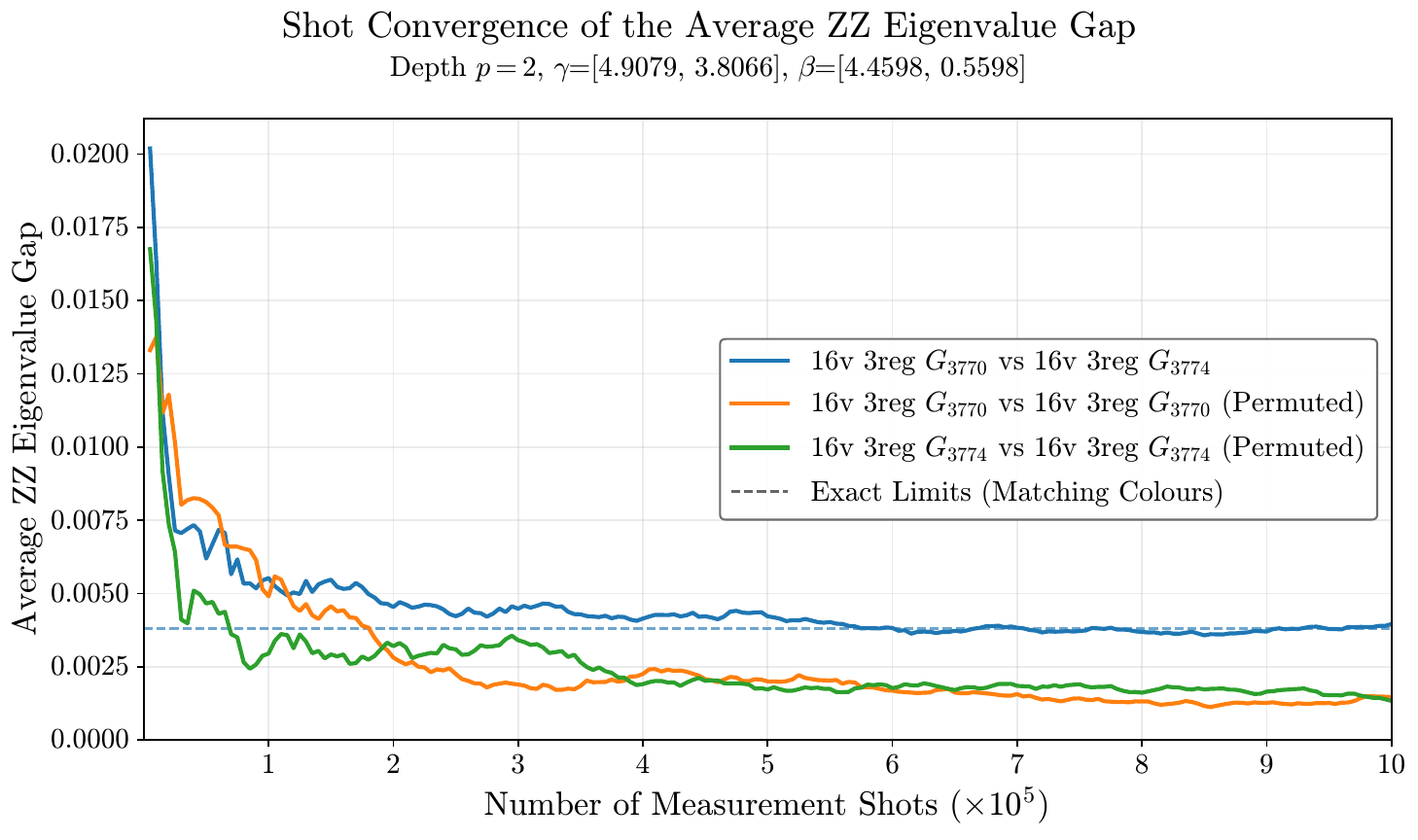}
    \caption{Shot convergence of $\Delta_{ZZ}$ for two $16$ vertex cospectral three-regular graphs.}
    \label{fig: Shot convergence of the average ZZ eigenvalue gap for G_{3770} vs G_{3774}}
\end{figure}

To test how well cospectral regular graphs perform in general, we calculated $\Delta_{ZZ}$ at 200 random parameter sequences for 10,486 cospectral graph pairs and averaged the results. Figure~\ref{fig: The scaling of Delta_ZZ for cospectral regular graphs} shows these results, with error bars indicating the Standard Error of the Mean (SEM).
\begin{figure}[tb!]
    \centering
    \includegraphics[width=\textwidth]{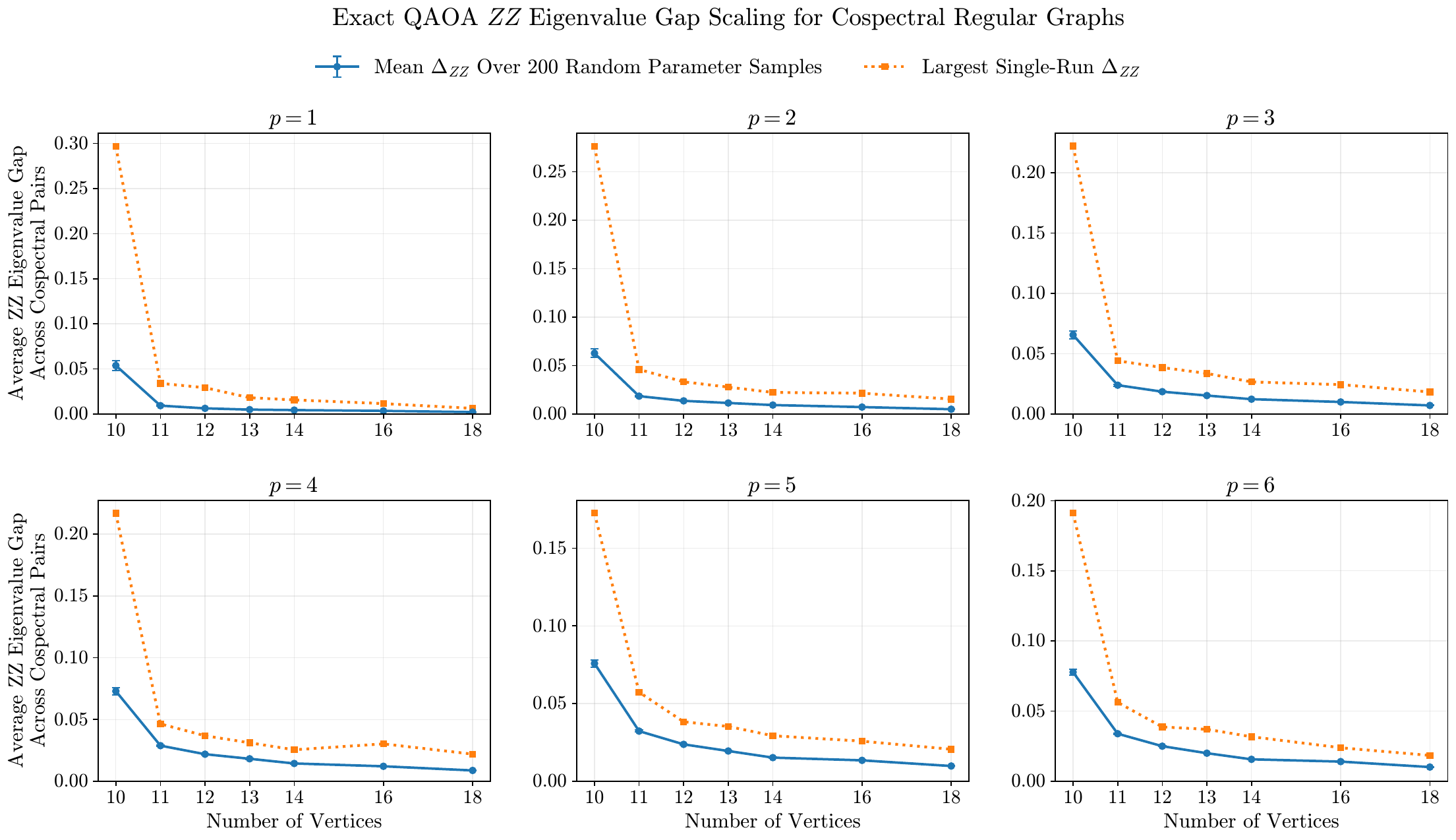}
    \caption{The scaling of $\Delta_{ZZ}$ for 10,486 cospectral regular graph pairs.}
    \label{fig: The scaling of Delta_ZZ for cospectral regular graphs}
\end{figure}

These numerical results motivate the following conjecture.

\begin{conjecture}{\textbf{ZZ-Spectrum Distinguishability}}\label{conj: ZZ-Spectrum Distinguishability}
    \\ For every pair of non-isomorphic graphs $G=(V,E)$ and $G'=(V,E')$, there exists a QAOA depth $p$ and parameters $\boldsymbol{\gamma},\boldsymbol{\beta} \in [0,2\pi]^p$ such that
    \[
    \operatorname{spec}(M_G) \neq \operatorname{spec}(M_{G'})
    \]
\end{conjecture}


\section{Testing graph isomorphism using IQP circuits}

{Now we move on to considering IQP circuits.} To test whether the Hamming weight distribution produced by the IQP circuit in Definition~\ref{def: IQP circuit} is useful for distinguishing graphs, we compared the full simulated Hamming weight distributions corresponding to all connected graphs with $|V|<8$ at 2000 equally spaced $\theta$ values in the interval $[0,\frac{\pi}{2}]$. This computation found that there exist pairs of 7-vertex graphs with Hamming weight distributions that are numerically indistinguishable for all 2000 values of $\theta$. This suggests that the Hamming weight distribution produced by the IQP circuit does not, in general, contain sufficient information to distinguish non-isomorphic graphs.

This limited distinguishing power is consistent with a structural restriction of the IQP circuits. Their output distributions are supported only on bitstrings of even Hamming weight. Furthermore, the Hamming weight distribution discards all information about how probability is distributed among individual bitstrings within each sector. We first establish the parity restriction.
 
\begin{lemma}{\textbf{The IQP circuit distribution has only even Hamming weight support}}\label{lemma: IQP support has only even Hamming weight}
\end{lemma}

\begin{proof}
    Let us start by writing the IQP circuit in the following form: 
    \begin{align*}
    \prod_{(i,j)\in E}e^{i\theta X_i X_j}\ket{0^n} &= H^{\otimes n}\left(\prod_{(i,j)\in E}e^{i\theta Z_i Z_j}\right)H^{\otimes n}\ket{0^n} \\ 
    &= \frac{1}{\sqrt{2^n}}\sum_{x\in\{0,1\}^n}H^{\otimes n}e^{i\theta\sum_{(i,j)\in E}(-1)^{x_i\oplus x_j}}\ket{x}\\
    &= \frac{1}{2^n}\sum_{x\in\{0,1\}^n}e^{i\theta\sum_{(i,j)\in E}(-1)^{x_i\oplus x_j}}\sum_{y\in\{0,1\}^n}(-1)^{x\cdot y}\ket{y}.
    \end{align*}
    Now group the sum over all bitstrings x into pairs $(x,\bar{x})$ where $\bar{x} = x \oplus (1,\dots, 1)$, we can show that cancellation occurs for all odd Hamming weight outputs.
    
    Since there are $2^n$ bitstrings, we always have an even number of bitstrings, so we can pair them together in this way.

    For every $(i,j)\in E$, $\bar{x}_i\oplus\bar{x}_j=x_i\oplus x_j$, so the graph-dependent phase factor is identical for $x$ and $\bar{x}$.

    Focusing on the term $(-1)^{x\cdot y}$, for each pair $(x,\bar{x})$, we get the contribution
    \begin{align*}
        (-1)^{x\cdot y} + (-1)^{\bar{x}\cdot y} &= (-1)^{x\cdot y} + (-1)^{(x \oplus (1,\dots, 1))\cdot y}\\
            &= (-1)^{x\cdot y} + (-1)^{x\cdot y}\cdot(-1)^{\bigoplus_{i}y_i}\\
            &= (-1)^{x\cdot y}(1 + (-1)^{\bigoplus_{i}y_i})\\
            &= \begin{cases}
                2\cdot (-1)^{x\cdot y} & \text{if the hamming weight of y is even}\\
                0 & \text{if the hamming weight of y is odd}
            \end{cases}
    \end{align*}
    Therefore, the bitstrings of odd Hamming weight do not contribute to the probability distribution for the IQP circuit, so we only need to consider bitstrings with even Hamming weight.
\end{proof}

However, the Hamming weight distribution does not retain any information regarding the individual probabilities within each Hamming weight group. We therefore introduce the following fine-grained invariant.

\begin{definition}{\textbf{Probability Multiplicity Within Hamming Weight Groups}}
\\ For a graph $G$, let $P_G:\{0,1\}^n\to[0,1]$ denote the probability distribution produced by its corresponding IQP or QAOA circuit. For each $k\in\{0,\dots,n\}$ and $p\in[0,1]$, define
\[
\mathcal{M}_{G}^k(p) \coloneqq \left|\{y\in\{0,1\}^n:|y|=k,\ P_G(y)=p\}\right|.
\]
Thus, $\mathcal{M}_{G}^k(p)$ is the number of bitstrings with Hamming weight $k$ that occur with probability $p$.
\end{definition}

\begin{restatable}{theorem}{iqpdistinvariant}\label{thrm: IQP dist invariant}
    \textbf{Hamming Weight Probability Multiplicity Invariant}
    \\ Let $P_G,P_{G'}:\{0,1\}^n\to[0,1]$ be the probability distributions produced by the IQP or QAOA circuits associated with graphs $G$ and $G'$, using identical circuit parameters.

If $G \cong G'$, then
\[
\mathcal{M}_{G}^k(p) = \mathcal{M}_{G'}^k(p) \qquad \forall k\in\{0,\dots,n\},\ p\in[0,1].
\]
Consequently, if there exist \(k\in\{0,\dots,n\}\) and \(p\in[0,1]\) such that
\[
\mathcal{M}_{G}^k(p) \neq \mathcal{M}_{G'}^k(p),
\]
then $G$ and $G'$ are non-isomorphic.

Moreover,
\[
\mathcal{M}_{G}^k(p) = \mathcal{M}_{G'}^k(p) \qquad \forall k\in\{0,\dots,n\},\ p\in[0,1]
\]
if and only if there exists a Hamming weight preserving bijection
\[
\phi:\{0,1\}^n\to\{0,1\}^n
\]
such that
\[
P_{G'}(y) = P_G(\phi(y))
\qquad
\forall y\in\{0,1\}^n.
\]
Hence, the probability multiplicities $\{\mathcal{M}_{G}^k(p)\}_{k,p}$ form a complete invariant under arbitrary Hamming weight preserving bijections of the circuit outcomes.
\end{restatable}

The proof is deferred to Appendix~\ref{appendix: Proof of HW probability multiplicity invariant}.

Every vertex permutation $\pi\in S_n$ induces a Hamming weight preserving bijection of bitstrings and is therefore a particular example of the map $\phi$ appearing in Theorem~\ref{thrm: IQP dist invariant}. However, a general Hamming weight preserving bijection $\phi$ need not be induced by a single permutation of the bit positions; it may rearrange bitstrings independently within each Hamming weight group. Thus, equality of the probability multiplicities implies only that the distributions are indistinguishable by this invariant. It does not imply that the graphs are isomorphic.

\begin{corollary}{\textbf{Probability Multiplicities Refine the Hamming Weight Distribution}}  \label{cor:probability_multiplicities_refine_HW}
\\For every graph $G$ and Hamming weight $k$,
\[
P_{G,HW}(k) = \sum_{\substack{p\in[0,1]\\\mathcal{M}_{G}^k(p)>0}} p\mathcal{M}_{G}^k(p).
\]
Consequently, the probability multiplicities $\{\mathcal{M}_{G}^k(p)\}_{k,p}$ determine the complete Hamming weight distribution.

In particular,
\[
\mathcal{M}_{G}^k(p) = \mathcal{M}_{G'}^k(p) \quad \forall k,p
\]
implies
\[
P_{G,HW}(k)=P_{G',HW}(k) \quad \forall k.
\]
The converse does not hold in general. Therefore, the probability multiplicities within Hamming weight groups form a strictly more fine-grained invariant of the output distribution than the Hamming weight distribution alone.
\end{corollary}

\begin{proof}
For a fixed $k$, the quantity $\mathcal{M}_{G}^k(p)$ counts the number of bitstrings of Hamming weight $k$ that each occur with probability $p$. Hence their total contribution to the probability of obtaining Hamming weight $k$ is $p\mathcal{M}_{G}^k(p)$. Summing over all probability values that occur within the group gives
\[
P_{G,HW}(k) = \sum_{\substack{p\in[0,1]\\\mathcal{M}_{G}^k(p)>0}} p\mathcal{M}_{G}^k(p).
\]

The converse does not hold because two Hamming weight groups can have the same total probability while containing different individual probability values.
\end{proof}

In our numerical experiments, probability multiplicities were computed using exact equality of the double-precision floating-point probability values. Under this numerical comparison, this fine-grained invariant distinguished all the non-isomorphic graphs considered in this project, using their full simulated IQP distributions, for at least one of the 2000 equally-spaced $\theta$ values in $[0,\frac{\pi}{2}]$. This evidence motivates the following conjecture.

\begin{conjecture}{\textbf{IQP Circuit Distinguishability}}\label{conj: IQP circuit distinguishability}
    \\For every pair of graphs $G=(V,E)$ and $G'=(V,E')$ where $G \ncong G'$, there exist some $\theta \in [0,\frac{\pi}{2}]$, $k\in\{0,\dots,n\}$ and $p\in[0,1]$ such that 
    \[
    \mathcal{M}_{G}^k(p) \neq \mathcal{M}_{G'}^k(p).
    \]
\end{conjecture}

Although this may provide a valid method for distinguishing non-isomorphic graphs, it is not obviously a practical one. The invariant in Theorem~\ref{thrm: IQP dist invariant} depends on the exact equality of individual outcome probabilities and their multiplicities. Therefore, directly evaluating it requires resolving the probabilities within each Hamming weight sector, potentially over an exponentially large set of outcomes. Moreover, exact equality is a numerically fragile comparison. Floating-point round-off can perturb probabilities that are theoretically identical in classical simulation, while finite-shot error and experimental noise introduce additional perturbations on quantum hardware.

Thus, while our numerical results provide evidence that IQP output distributions may contain enough information to distinguish non-isomorphic graphs, this invariant does not by itself yield a practical quantum graph isomorphism algorithm.

We do not state the analogous conjecture for QAOA circuits, since Conjecture~\ref{conj: QAOA Hamming weight distinguishability} already proposes that the coarser Hamming weight distribution is sufficient to distinguish non-isomorphic graphs.


\section{Analysis of IQP circuits for specific parameters}

Although IQP circuits may provide worse distinguishing power than QAOA circuits, we can study how their distributions encode graph information at specific values of $\theta$.

\begin{lemma}\label{lemma: IQP dist at pi/2}
    For a graph $G=(V,E)$, at $\theta=\pi/2$, the graph's corresponding IQP circuit distribution becomes
\[
    p(y) = |\bra{y} \prod_{k = 1}^n ( X_k)^{\deg(k)}\ket{0^n}|^2
      =  \begin{cases}
        1 & \text{if } y_k = \deg(k) \pmod 2 \quad \forall k \in \{1,\dots,n\} \\
        0 & \text{otherwise}
    \end{cases}
\]
for $y\in\{0,1\}^n$, where $\deg(k)$ denotes the degree of vertex $k$.
\end{lemma}
\begin{proof}
    Using the fact that for a unitary, Hermitian operator $U$, $e^{i\theta U} = \cos{\theta}I + i\sin{\theta} U$,
    \begin{align*}
        p(y)&=|\bra{y}\prod_{(i,j) \in E}e^{i\frac{\pi}{2}X_i X_j}\ket{0^n}|^2
        \\ &=|\bra{y}\prod_{(i,j) \in E}\cos{\frac{\pi}{2}}I + i\sin{\frac{\pi}{2}} (X_i X_j) \ket{0^n}|^2
        \\ &=|\bra{y}\prod_{(i,j) \in E}i (X_i X_j) \ket{0^n}|^2
        \\ &= |\bra{y} i^{|E|} \prod_{k = 1}^n ( X_k)^{\deg(k)}\ket{0^n}|^2
        \\ &= |\bra{y} \prod_{k = 1}^n ( X_k)^{\deg(k)}\ket{0^n}|^2
        \\ &=  \begin{cases}
        1 & \text{if } y = y' \text{, where } y'_k = \deg(k) \pmod 2 \\
        0 & \text{otherwise}
    \end{cases}
    \end{align*}
\end{proof}

\begin{restatable}{lemma}{iqpbecomesclifford}\label{lemma: IQP becomes Clifford at pi/4}
    When $\theta = \frac{\pi}{4}$, the IQP circuit's distribution can be efficiently simulated classically.
\end{restatable}

The proof is deferred to Appendix~\ref{appendix: proof that IQP becomes Clifford at pi/4}

{Next, we use this to study the form of the resulting probability distributions. Define} the degree-parity matrix
\[
D= \operatorname{diag}(\deg(1)\bmod 2,\deg(2)\bmod 2,\dots,\deg(n)\bmod 2)\in\mathbb{F}_2^{n\times n}.
\]
Thus, whenever $A+D$ is used below, the matrix arithmetic is over $\mathbb{F}_2$.

Define further the $\mathbb{Z}_4$-valued quadratic form
    \[
    q(x)=2\sum_{1\leq i<j\leq n}A_{ij}x_i x_j+3\sum_{k=1}^{n}\deg(k)x_k \pmod 4,
    \]
    and the subspace
    \[
    \mathcal{L}_q=\{x\in\mathbb{F}_2^n:q(x\oplus z)\equiv q(x)+q(z)\pmod 4 \hspace{3mm}\forall z\in\mathbb{F}_2^n\}.
    \]

\begin{lemma}\label{lemma: IQP dist at pi/4}
    When measuring the IQP circuit \ref{def: IQP circuit} with $\theta=\pi/4$, the probability of obtaining a bitstring $y\in\mathbb{F}_2^n$ is non-zero, $P(y)>0$, if and only if
    \[
    q(x)\equiv2y^Tx\pmod4 \hspace{3mm}\forall x\in\mathcal{L}_q.
    \]
    Furthermore, $P(y)$ is uniform on the set of all solutions $y$ and $\mathcal{L}_q=\ker_{\mathbb{F}_2}(A+D).$

    More explicitly, let
    \[
    r=\operatorname{rank}_{\mathbb{F}_2}(A+D)
    \]
    and let $y_0$ be any solution to
    \[
    q(x)\equiv2y_0^Tx\pmod4 \hspace{3mm}\forall x\in\mathcal{L}_q.
    \]
    Then the support of the distribution is the affine subspace
    \[
    \operatorname{supp}(P)=y_0\oplus\operatorname{im}_{\mathbb{F}_2}(A+D),
    \]
    and the complete probability distribution is
    \[
    P(y)=
    \begin{cases}
        2^{-r} & \text{if } y\in y_0\oplus\operatorname{im}_{\mathbb{F}_2}(A+D),\\
        0 & \text{otherwise}.
    \end{cases}
    \]
    Hence, the IQP circuit distribution at $\theta=\pi/4$ has exactly $2^r$ bitstrings in its support.
\end{lemma}

\begin{proof}
    Starting with the expression for the IQP circuit obtained in the proof of Lemma~\ref{lemma: IQP support has only even Hamming weight}, for $y\in\{0,1\}^n$,
    \begin{align*}
        P(y) &= \left|\frac{1}{2^n}\sum_{x\in\{0,1\}^n}e^{i\theta\sum_{(i,j)\in E}(-1)^{x_i\oplus x_j}}(-1)^{x\cdot y}\right|^2\\
        &= \left|\frac{1}{2^n}\sum_{x\in\{0,1\}^n}e^{i\theta\sum_{(i,j)\in E}(1-2x_i)(1-2x_j)}(-1)^{x\cdot y}\right|^2\\
        &= \left|\frac{e^{i\theta|E|}}{2^n}\sum_{x\in\{0,1\}^n}e^{i\theta\left(-2\sum_{k=1}^{n}\deg(k)x_k+4\sum_{(i,j)\in E}x_ix_j\right)}(-1)^{x\cdot y}\right|^2.
    \end{align*}
    The factor $e^{i\theta|E|}$ is a global phase and therefore disappears when taking the absolute value squared.

    Substituting in $\theta=\frac{\pi}{4}$,
    \begin{align*}
        P(y) &= \left|\frac{1}{2^n}\sum_{x\in\{0,1\}^n}e^{-\frac{i\pi}{2}\sum_{k=1}^{n}\deg(k)x_k}e^{i\pi\sum_{(i,j)\in E}x_ix_j}(-1)^{x\cdot y}\right|^2\\
        &= \left|\frac{1}{2^n}\sum_{x\in\{0,1\}^n}i^{3\sum_{k=1}^{n}\deg(k)x_k+2\sum_{(i,j)\in E}x_ix_j+2x\cdot y}\right|^2\\
        &= \left|\frac{1}{2^n}\sum_{x\in\{0,1\}^n}i^{2\sum_{1\leq i<j\leq n}A_{ij}x_ix_j+3\sum_{k=1}^{n}\deg(k)x_k+2x\cdot y}\right|^2.
    \end{align*}

    Therefore, using the quadratic form $q(x)$ defined above,
    \[
    P(y)=\left|\frac{1}{2^n}\sum_{x\in\mathbb{F}_2^n}i^{q(x)}(-1)^{x\cdot y}\right|^2.
    \]
    This manipulation allows us to match the distribution of our specific IQP circuit to the Hidden Linear Function distribution considered by Bravyi, Gosset and K\"onig in their paper `Quantum advantage with shallow circuits'~\cite{Bravyi_2018}.

    To determine the subspace $\mathcal{L}_q$, note that
    \[
    q(x\oplus z)\equiv q(x)+q(z)+2z^T(A+D)x\pmod4.
    \]

    Thus $x\in\mathcal{L}_q$ if and only if
    \[
    2z^T(A+D)x\equiv0\pmod4 \hspace{3mm}\forall z\in\mathbb{F}_2^n.
    \]
    Equivalently,
    \[
    z^T(A+D)x=0\pmod2 \hspace{3mm}\forall z\in\mathbb{F}_2^n.
    \]
    This holds for every $z$ if and only if
    \[
    (A+D)x=0 \pmod2.
    \]
    Hence,
    \[
    \mathcal{L}_q=\ker_{\mathbb{F}_2}(A+D).
    \]

    We can now apply Lemma 2 of Bravyi, Gosset and K\"onig~\cite{Bravyi_2018}, which gives
    \[
    P(y)>0 \quad\Longleftrightarrow\quad q(x)\equiv2y^Tx\pmod4 \hspace{3mm}\forall x\in\mathcal{L}_q,
    \]
    and also shows that the probability distribution is uniform over the set of all such solutions $y$.

    We can make the form of this uniform distribution explicit. Let $y_0$ be any solution. Another bitstring $y$ is also a solution if and only if
    \[
    2y^Tx\equiv2y_0^Tx\pmod4 \hspace{3mm}\forall x\in\mathcal{L}_q,
    \]
    which is equivalent to
    \[
    (y\oplus y_0)^Tx=0\pmod2 \hspace{3mm}\forall x\in\mathcal{L}_q.
    \]
    Thus,
    \[
    y\oplus y_0\in\mathcal{L}_q^\perp,
    \]
    and consequently
    \[
    \operatorname{supp}(P)=y_0\oplus\mathcal{L}_q^\perp.
    \]

    Since $A+D$ is symmetric,
    \[
    \mathcal{L}_q^\perp=\left(\ker_{\mathbb{F}_2}(A+D)\right)^\perp=\operatorname{im}_{\mathbb{F}_2}\left((A+D)^T\right)=\operatorname{im}_{\mathbb{F}_2}(A+D).
    \]
    Therefore,
    \[
    \operatorname{supp}(P)=y_0\oplus\operatorname{im}_{\mathbb{F}_2}(A+D).
    \]

    Let
    \[
    r=\operatorname{rank}_{\mathbb{F}_2}(A+D).
    \]
    The image of $A+D$ is an $r$-dimensional vector space over $\mathbb{F}_2$, and therefore contains exactly $2^r$ elements. Hence the support of $P$ contains exactly $2^r$ bitstrings. Since Lemma 2 of~\cite{Bravyi_2018} shows that the distribution is uniform over these bitstrings, each supported bitstring has probability $2^{-r}$.
    We therefore obtain the full distribution
    \[
    P(y)=
    \begin{cases}
        2^{-r} & \text{if } y\in y_0\oplus\operatorname{im}_{\mathbb{F}_2}(A+D),\\
        0 & \text{otherwise}.
    \end{cases}
    \]
\end{proof}


\section{Conclusion and outlook}

In this work, we have investigated how much information about a graph can be extracted from the distributions of simple graph-dependent IQP and QAOA circuits. We first showed that suitable IQP and QAOA constructions contain enough information to reconstruct the underlying graph using a polynomial number of samples. Thus, although coarse statistics of the output distributions may fail to distinguish particular graph pairs, the distributions themselves do not fundamentally discard the adjacency information required to determine graph isomorphism.

Our numerical results show that the QAOA energy is an effective distinguishing statistic for many graph pairs, consistent with the observations of Szegedy~\cite{szegedy_2019_qaoa}. However, we also identified graph pairs for which additional information in the distribution provides a clear advantage at low circuit depth. In particular, for circular and M\"obius ladder graphs on $16$ vertices, the numerically computed value of $|\Delta E|$ is exactly zero for $p<4$, while for the corresponding $20$-vertex pair it is exactly zero for $p<5$. In both cases, the computed Hamming weight distributions are distinct for every depth tested. This demonstrates that graph-distinguishing information can be present in the QAOA output distribution before it becomes visible in the expectation value of the cost Hamiltonian.

We also investigated the matrix of two-qubit correlators $M_G$. Under a graph isomorphism this matrix undergoes a simultaneous permutation of its rows and columns, so its spectrum provides a graph invariant. Numerically, this invariant proved highly discriminating for cospectral regular graphs. At depth $p=2$, all 10,486 cospectral regular graph pairs considered in our experiments were distinguished by their ZZ correlation spectrum. We further showed that the invariant remains accessible when the correlation matrix is estimated from a finite number of measurements. If the exact average ZZ eigenvalue gap is $\delta>0$, then $O(n^2\log n/\delta^2)$ samples are sufficient to estimate the gap to within a fixed relative error with high probability. This provides an explicit connection between the exact graph invariant and the resources required to resolve it experimentally.

Our results for IQP circuits illustrate both the limitations and the structure of their output distributions. The Hamming weight distribution alone does not distinguish all non-isomorphic graphs even among the small instances considered here. However, a finer invariant, based on the multiplicities of individual probabilities within each Hamming weight sector, distinguished every non-isomorphic pair considered in our numerical experiments. We showed that this quantity is a complete invariant of the output distribution under arbitrary Hamming weight-preserving bijections. At particular circuit parameters, the IQP distribution also admits a direct graph-theoretic interpretation. At $\theta=\pi/2$, the output is determined by the parity of the vertex degrees, while at $\theta=\pi/4$, the distribution is uniform over an affine subspace $y_0\oplus\operatorname{im}_{\mathbb{F}_2}(A+D)$, where $y_0$ is a solution to the Hidden Linear Function problem.

These results leave open the central question of whether any of these approaches can lead to an efficient graph-isomorphism algorithm. In particular, it is necessary to understand how the depth and statistical separation required to distinguish difficult graph pairs scale with the number of vertices. 

In our investigation, we used randomly chosen parameters rather than optimising the circuits for distinguishability. An important direction for future work is therefore to investigate whether more effective parameter-selection strategies can improve distinguishability. Barren plateaus are commonly characterised by exponentially vanishing gradients~\cite{Larocca_2022}, but they are also closely associated with concentration of cost-function values around their mean~\cite{McClean2018BarrenPlateaus,Cerezo2021BarrenPlateaus,Arrasmith2022CostConcentration}. In the context of QAOA specifically, Brand\~ao \emph{et al.} showed that, for fixed parameters, the cost concentrates across typical problem instances, proving such concentration for low-depth MaxCut QAOA on large 3-regular graphs~\cite{Brandao2018QAOAConcentration}. This does not imply that the low-depth QAOA circuits considered here necessarily exhibit barren plateaus. Indeed, for certain shallow parametrised circuits, local cost functions have been shown to avoid exponentially vanishing gradients under appropriate assumptions~\cite{Cerezo2021BarrenPlateaus}. However, concentration across graph instances is potentially problematic for graph isomorphism, since distinguishability relies precisely on differences between graph-dependent statistics. Therefore, it is possible that, as the graph size increases, different graphs produce increasingly similar values of the statistics considered here, causing quantities such as the energy gap, Hamming weight TVD or ZZ eigenvalue gap to become exponentially small. This would lead to a corresponding increase in the number of measurements required to distinguish the graphs, potentially preventing an efficient graph-isomorphism algorithm even at low circuit depth. Determining whether, and how rapidly, such concentration occurs for the circuit families and graph invariants considered here is therefore an important direction for future work.

The exponential cost of full statevector simulation currently limits a direct numerical investigation of this scaling. One natural next step is therefore to study substantially larger graph instances using simulation methods that do not require construction of the complete output distribution. In particular, Pauli-propagation methods offer a promising route for estimating expectation values of Pauli observables in larger circuits~\cite{Rudolph2025PauliPropagation}. Such methods are particularly well suited to quantities such as the QAOA energy and low-order correlators, and could allow the scaling of the ZZ-based invariants to be explored at sizes beyond those accessible here. The Hamming weight distribution contains more global information and may require different approximation or sampling techniques.

There are also many possible graph invariants beyond those studied here. Higher-order correlators, invariants derived from larger correlation tensors, or other permutation-invariant functions of QAOA observables may retain graph information not captured by the energy, Hamming weight distribution, or ZZ spectrum. These could be tested systematically on larger collections of cospectral and strongly regular graphs, as well as on graph families specifically constructed to be difficult for classical graph-isomorphism and refinement procedures. Such experiments could help identify which structural properties influence the circuit depth required for distinguishability via each invariant.

Finally, an important next step is to move beyond classical simulation and test these invariants on quantum hardware. The low circuit depths at which several of the distinguishing effects observed here occur make this a natural setting in which to study the combined effects of finite sampling, gate noise and hardware connectivity. In particular, experimental measurements of the ZZ correlation spectrum would provide a test of the finite-shot analysis developed here, while comparisons with exact small-instance results would allow the robustness of the proposed invariants to realistic noise to be quantified.

More broadly, the results of this work suggest that simple graph-dependent quantum circuits can encode substantially more structural information than is visible through a single expectation value. Determining whether this information can be extracted with polynomial resources for increasingly difficult graph families is the key question for establishing whether such circuits can provide a useful quantum approach to graph isomorphism.


\subsection*{Acknowledgments}

We would like to thank Freddie Burns for helpful preliminary work on using IQP circuits to test graph isomorphism, and for guidance on efficient implementation of IQP circuits, and would also like to thank Mario Szegedy for comments on a previous version.

\subsection*{Statement on use of AI}

We have used LLMs for finding relevant references,  feedback on our text, to help improve the clarity of the exposition, to check correctness of results, as well as for some coding tasks. The authors take full responsibility for the correctness of all parts of the paper.


\bibliographystyle{plain}
\bibliography{references}


\begin{appendices}
    
\section{Proof of Lemma~\ref{lemma: IQP permuting nodes permutes bits}}\label{appendix: proof of IQP permuting nodes permutes bits}
\iqppermutationlemma* 

\begin{proof}
    Let $U_\pi$ be the unitary operator that permutes qubits according to the permutation $\pi$, so that
    \[
    U_\pi X_i U_\pi^\dagger = X_{\pi(i)}.
    \]
    Since 
    \[
    \prod_{(i,j) \in E}\exp\left(i\theta X_i X_j\right)\ket{0^n}
    \]
    is the IQP circuit that corresponds to the graph $G$, the IQP circuit for the graph $G'$, $\ket{\text{IQP}_{G'}}$, is
    \begin{align*}  
    \prod_{(i,j) \in E}\exp\left(i\theta X_{\pi(i)} \otimes  X_{\pi(j)}\right)\ket{0^n}&=\prod_{(i,j) \in E}\exp\left(i\theta U_\pi( X_i \otimes  X_j)U_\pi^\dagger\right)\ket{0^n}.\\
    &=\prod_{(i,j) \in E}U_\pi\exp\left(i\theta( X_i \otimes  X_j)\right)U_\pi^\dagger\ket{0^n}\\
    &=U_\pi\left(\prod_{(i,j) \in E}\exp\left(i\theta( X_i \otimes  X_j)\right)\right)U_\pi^\dagger\ket{0^n}\\
    &=U_\pi\prod_{(i,j) \in E}\exp\left(i\theta( X_i \otimes  X_j)\right)\ket{0^n}\\
    &=U_\pi \ket{\text{IQP}_G},
    \end{align*}
    where we have used the fact that any permutation matrix is unitary, $\ket{0^n}$ is invariant under permutations, and for a unitary operator U and any operator Y, 
    \[
    e^{UYU^\dagger} = Ue^YU^\dagger.
    \]
    Also note that we still index the product over $(i,j)\in E$ since we apply the permutation $\pi$ to the operators $X_i$ so that $X_i$ acts on qubit $\pi(i)$. \vspace{2mm}\\
    Finally, for any $x \in \mathbb{F}_2^n$,
    \begin{align*}
        |\langle x | \text{IQP}_{G'} \rangle|^2
    &= |\langle x | U_\pi|\text{IQP}_G \rangle|^2\\
    &=  |\langle \pi^{-1}(x) | \text{IQP}_G \rangle|^2,
    \end{align*}
    so the output distribution is permuted accordingly.
\end{proof}

\section{Proof of Lemma~\ref{lemma: QAOA permuting nodes permutes bits}}\label{appendix: proof of QAOA permuting nodes permutes bits}

\qaoapermutationlemma*

\begin{proof}
    Let $U_\pi$ be the unitary operator that permutes qubits so that
    \[
    X_{\pi(i)} = U_\pi X_i U_\pi^\dagger, \quad Z_{\pi(i)} =U_\pi Z_i U_\pi^\dagger.
    \]
    First, consider the cost Hamiltonian terms. For $(j,k)\in E$,
    \[
    C_{(\pi(j),\pi(k))} = \tfrac{1}{2}(I + Z_{\pi(j)} Z_{\pi(k)}) = U_\pi C_{(j,k)} U_\pi^\dagger.
    \]
    Similarly, for the mixer terms,
    \[
    X_{\pi(i)} = U_\pi X_i U_\pi^\dagger.
    \]
    Using the fact that for a unitary operator U and any operator Y, $e^{UYU^\dagger} = Ue^YU^\dagger$,
    \[
    e^{-i\gamma_q C_{(\pi(j),\pi(k))}} = U_\pi e^{-i\gamma_q C_{(j,k)}} U_\pi^\dagger, 
    \quad
    e^{-i\beta_q X_{\pi(i)}} = U_\pi e^{-i\beta_q X_i} U_\pi^\dagger.
    \]
    Thus the QAOA circuit for $G'$ is
    \begin{align*}
    \ket{\text{QAOA}_{G'}} 
    &= \prod_{q=0}^{p-1}\left(\prod_{i\in V}e^{-i\beta_q X_{\pi(i)}} \prod_{(j,k)\in E} e^{-i\gamma_q C_{(\pi(j),\pi(k))}} \right) H^{\otimes n} \ket{0^n} \\
    &= \prod_{q=0}^{p-1}\left(\prod_{i\in V}U_\pi e^{-i\beta_q X_i}U_\pi^\dagger \prod_{(j,k)\in E} U_\pi e^{-i\gamma_q C_{(j,k)}} U_\pi^\dagger \right) H^{\otimes n} \ket{0^n}.
    \end{align*}
    We simplify using the fact that $U_\pi$ is unitary and that $H^{\otimes n}\ket{0^n}$ is the equal superposition over all states, so it is invariant under permutations (i.e. $U_\pi^\dagger H^{\otimes n}\ket{0^n} = H^{\otimes n}\ket{0^n}$).
    \begin{align*}
        \ket{\text{QAOA}_{G'}} 
    &= U_\pi \left[
    \prod_{q=0}^{p-1}\left(\prod_{i\in V} e^{-i\beta_q X_i} \prod_{(j,k)\in E} e^{-i\gamma_q C_{(j,k)}} \right)
    \right] U_\pi^\dagger H^{\otimes n} \ket{0^n}\\
    &= U_\pi \left[
    \prod_{q=0}^{p-1}\left(\prod_{i\in V} e^{-i\beta_q X_i} \prod_{(j,k)\in E} e^{-i\gamma_q C_{(j,k)}} \right)
    \right] H^{\otimes n} \ket{0^n}\\
    &= U_\pi \ket{\text{QAOA}_G}.
    \end{align*}
    Lastly, for any $x \in \mathbb{F}_2^n$,
    \begin{align*}
    |\langle x | \text{QAOA}_{G'} \rangle|^2
    &= |\langle x | U_\pi | \text{QAOA}_G \rangle|^2 \\
    &= |\langle \pi^{-1}(x) | \text{QAOA}_G \rangle|^2,
    \end{align*}
    so the output distribution is permuted accordingly.
\end{proof}

\section{Proof of Theorem~\ref{thrm: IQP dist invariant}}\label{appendix: Proof of HW probability multiplicity invariant}

\iqpdistinvariant*

\begin{proof}
Suppose that $G\cong G'$. By Lemma~\ref{lemma: IQP permuting nodes permutes bits} for IQP circuits, and Lemma~\ref{lemma: QAOA permuting nodes permutes bits} for QAOA circuits, there exists a permutation $\pi\in S_n$ such that
\[
P_{G'}(y) = P_G(\pi(y)) \qquad \forall y\in\{0,1\}^n.
\]
The Hamming weight is invariant under the permutation $\pi$. Furthermore, since $\pi$ is a bijection, it preserves the number of occurrences of each probability value within each Hamming weight group. Therefore,
\[
\mathcal{M}_{G}^k(p) = \mathcal{M}_{G'}^k(p) \qquad \forall k\in\{0,\dots,n\},\ p\in[0,1].
\]
Taking the contrapositive proves that a difference between the probability multiplicities is sufficient to identify $G$ and $G'$ as non-isomorphic.

It remains to prove the final equivalence. Suppose first that
\[
\mathcal{M}_{G}^k(p) = \mathcal{M}_{G'}^k(p) \qquad \forall k\in\{0,\dots,n\},\ p\in[0,1].
\]
Fix $k$ and define
\[
\mathcal{Y}_k = \{y\in\{0,1\}^n:|y|=k\}.
\]
Since the probability multiplicities agree within $\mathcal{Y}_k$,
there exists a bijection
\[
\phi_k:\mathcal{Y}_k\to\mathcal{Y}_k
\]
such that
\[
P_{G'}(y)=P_G(\phi_k(y)) \qquad \forall y\in\mathcal{Y}_k.
\]
Define a global bijection
\[
\phi:\{0,1\}^n\to\{0,1\}^n
\]
by setting
\[
\phi(y)=\phi_k(y) \qquad \text{whenever }|y|=k.
\]
Since each $\phi_k$ maps $\mathcal{Y}_k$ onto itself, $\phi$ preserves Hamming weight. Moreover,
\[
P_{G'}(y)=P_G(\phi(y)) \qquad \forall y\in\{0,1\}^n.
\]

Conversely, suppose that there exists a Hamming weight preserving bijection $\phi$ satisfying
\[
P_{G'}(y)=P_G(\phi(y)) \qquad \forall y\in\{0,1\}^n.
\]
For every $k$, the map $\phi$ restricts to a bijection from $\mathcal{Y}_k$ onto itself. It therefore preserves the number of occurrences of each probability value within $\mathcal{Y}_k$, giving
\[
\mathcal{M}_{G}^k(p) = \mathcal{M}_{G'}^k(p) \qquad \forall k\in\{0,\dots,n\},\ p\in[0,1].
\]
This proves the result.
\end{proof}

\section{Proof of Lemma~\ref{lemma: IQP becomes Clifford at pi/4}}\label{appendix: proof that IQP becomes Clifford at pi/4}

\iqpbecomesclifford*

\begin{proof}
    The aim is to show that when $\theta = \frac{\pi}{4}$, $e^{i\theta X_i X_j}$ becomes a Clifford gate, so by the Gottesman-Knill Theorem~\cite{Gottesman_1998_GK_theorem}, the IQP circuit can be efficiently simulated classically.
    
    To show that $e^{i\frac{\pi}{4} X_i X_j}$ is a Clifford gate, it suffices to show that it is generated by $H$, $S$ and $CNOT$, since they are the generators of the Clifford group~\cite{Gottesman_1998_Clifford}.
        
    Note that $X=HZH$ and for a unitary operator $U$ and any operator $Y$, 
    \[
    e^{UYU^\dagger} = Ue^YU^\dagger.
    \]
    Applying these facts, we obtain
    \[
    e^{i\frac{\pi}{4} X_i X_j} = (H_i \otimes H_j) e^{i\frac{\pi}{4} Z_i Z_j} (H_i \otimes H_j).
    \]
    So all that remains is to show that $e^{i\frac{\pi}{4} Z \otimes Z}$ is a Clifford gate.
    
    Firstly, $Z\otimes Z = \operatorname{diag}(1,-1,-1,1)$ so, up to an overall phase,
    \begin{align*}
        e^{i\frac{\pi}{4} Z \otimes Z} &= \operatorname{diag}(e^{i\frac{\pi}{4}},e^{-i\frac{\pi}{4}},e^{-i\frac{\pi}{4}},e^{i\frac{\pi}{4}})\\
        &=e^{i\frac{\pi}{4}}\operatorname{diag}(1,-i,-i,1)\\
        &\equiv \operatorname{diag}(1,-i,-i,1)
    \end{align*}
    Note that $S^\dagger \otimes S^\dagger = \operatorname{diag}(1,-i,-i,-1)$ and $C(Z)= \operatorname{diag}(1,1,1,-1).$ Thus,
    \[
        e^{i\frac{\pi}{4} Z \otimes Z} \equiv (S^\dagger \otimes S^\dagger)C(Z).
    \]
    $C(Z) = (I \otimes H)CNOT(I \otimes H),$ and $S^\dagger=S^3$, so the gate $e^{i\frac{\pi}{4} X_i X_j}$ is indeed generated by \(H,S,\text{ and }\mathrm{CNOT}\).
\end{proof}

\end{appendices}
\end{document}